\pdfoutput=1

\def\masterthesis{0}
\def\draft{0}
\def\sigconf{0}
\def\big{0}
\def\anon{0}
\def\shownomenclature{0}
\ifnum\big=1

    \documentclass[final,17pt]{extarticle} \usepackage{fullpage}

\else
    \ifnum\sigconf=1
        \documentclass[sigconf,final]{acmart}
    \else
        \ifnum\draft=1
           \documentclass[11pt,draft,a4paper]{article}
        \else
            \documentclass[11pt,final,a4paper]{article}
        \fi
        \usepackage{lmodern}
    \fi
\fi

\input{macros}

\usepackage[nottoc]{tocbibind} 

\usepackage[inline]{enumitem}
\usepackage{xspace,textcomp}
\usepackage[makeroom]{cancel}
\DeclareMathAlphabet{\mathpzc}{OT1}{pzc}{m}{it}
\newcommand\visiblespace{\text{\textvisiblespace}}

\newcommand{\eur}{\text{\texteuro}}
\newcommand{\str}{\mathsterling}

\newcommand{\crs}{\ensuremath{\mathsf{crs}}\xspace}
\newcommand{\account}{\ensuremath{\mathsf{account}}\xspace}
\newcommand{\permittedaccount}{\ensuremath{\mathsf{permittedaccount}}\xspace}

\newcommand{\caccount}{\ensuremath{\mathsf{caccount}}\xspace}
\newcommand{\forwardaccount}{\ensuremath{\mathsf{forwardaccount}}\xspace}

\newcommand{\mpk}{\ensuremath{\mathsf{mpk}}\xspace}
\newcommand{\msk}{\ensuremath{\mathsf{msk}}\xspace}
\newcommand{\bpk}{\ensuremath{\mathsf{bpk}}\xspace}
\newcommand{\bsk}{\ensuremath{\mathsf{bsk}}\xspace}

\newcommand{\csk}{\ensuremath{\mathsf{csk}}\xspace}
\newcommand{\cvk}{\ensuremath{\mathsf{cvk}}\xspace}

\newcommand{\intersectingfamily}{\ensuremath{\mathsf{intersectingFamily}}\xspace}

\newcommand{\aux}{\ensuremath{\mathsf{aux}}\xspace}
\newcommand{\interpretedmessage}{\ensuremath{\mathsf{interpretedmessage}}\xspace}

\newcommand{\witness}{\ensuremath{\mathsf{witness}}\xspace}
\newcommand{\vswitness}{\ensuremath{\mathsf{vswitness}}\xspace}
\newcommand{\pwit}{\ensuremath{\mathsf{pwit}}\xspace}

\newcommand{\rand}{\ensuremath{\mathsf{rand}}\xspace}
\newcommand{\dividends}{\ensuremath{\mathsf{dividends}}\xspace}

\newcommand{\balance}{\ensuremath{\mathsf{Balance}}\xspace}
\newcommand{\verifycolored}{\ensuremath{\mathsf{VerifyColored}}\xspace}
\newcommand{\coloraccount}{\ensuremath{\mathsf{Color}}\xspace}
\newcommand{\outputscript}{\ensuremath{\mathsf{OutputScript}}\xspace}
\newcommand{\outputverifyscript}{\ensuremath{\mathsf{OutputVerifyScript}}\xspace}
\newcommand{\restrictedoutputverifyscript}{\ensuremath{\mathsf{RestrictedOutputVerifyScript}}\xspace}
\newcommand{\siginterpreter}{\ensuremath{\mathsf{SigInterpreter}}\xspace}
\newcommand{\simplesiginterpreter}{\ensuremath{\mathsf{SimpleSigInterpreter}}\xspace}
\newcommand{\multisiginterpreter}{\ensuremath{\mathsf{MultisigInterpreter}}\xspace}
\newcommand{\restrictedsiginterpreter}{\ensuremath{\mathsf{RestrictedSigInterpreter}}\xspace}

\newcommand{\exchangeoutputverifyscript}{\ensuremath{\mathsf{ExOVS}}\xspace}

\newcommand{\keygen}{\ensuremath{\mathsf{Gen}}\xspace}
\newcommand{\setup}{\ensuremath{\mathsf{Setup}}\xspace}
\newcommand{\coloredcoinkeygen}{\ensuremath{\mathsf{ColoredCoinKeyGen}}\xspace}

\renewcommand{\verify}{\ensuremath{\mathsf{Verify}}} 
\providecommand{\sign}{\ensuremath{\mathsf{Sign}}\xspace}

\newcommand{\topup}{\ensuremath{\mathsf{TopUp}}\xspace}
\newcommand{\optopup}{\texttt{OPTOPUP}\xspace}

\newcommand{\simpleoutputscript}{\ensuremath{\mathsf{SimpleOutputScript}}\xspace}
\newcommand{\permanentoutputscript}{\ensuremath{\mathsf{PermanentOutputScript}}\xspace}

\newcommand{\coloredoutputscript}{\ensuremath{\mathsf{ColoredOutputScript}}\xspace} 

\newcommand{\simpleoutputverifyscript}{\ensuremath{\mathsf{SimpleOutputVerifyScript}}\xspace}
\newcommand{\alg}{\ensuremath{\mathsf{Alg}}\xspace}
\newcommand{\sen}{\ensuremath{\mathsf{Sen}}\xspace}
\newcommand{\out}{\ensuremath{\mathsf{OUT}}\xspace}
\renewcommand{\ds}{\ensuremath{\mathsf{DS}}\xspace}
\newcommand{\ts}{\ensuremath{\mathsf{TS}}\xspace}
\newcommand{\tsms}{\ensuremath{\mathsf{TMS}}\xspace}
\newcommand{\qs}{\ensuremath{\mathsf{US}}\xspace}

\newcommand{\oss}{\ensuremath{\mathsf{OSS}}\xspace}

\newcommand{\qsk}{\ensuremath{\mathpzc{sk}}\xspace}

\newcommand{\qsign}{\ensuremath{\mathpzc{Sign}}\xspace}
\newcommand{\qalg}{\ensuremath{\mathpzc{Alg}}\xspace}
\newcommand{\createsimpleaccount}{\ensuremath{\mathpzc{CreateSimpleAccount}}\xspace}

\newcommand{\createmultisigaccount}{\ensuremath{\mathpzc{CreateMultisigAccount}}\xspace}
\newcommand{\createrestrictedaccount}{\ensuremath{\mathpzc{CreateRestrictedAccount}}\xspace}

\newcommand{\rec}{\ensuremath{\mathpzc{Rec}}\xspace}
\newcommand{\qpt}{\ensuremath{\mathsf{QPT}}\xspace}

\newcommand{\tokengen}{\ensuremath{\mathpzc{Gen}}\xspace}

\newcommand{\qgen}{\ensuremath{\mathpzc{Gen}}\xspace}

\newcommand{\qadv}{\ensuremath{\mathpzc{Adv}}\xspace}
\newcommand{\simplepay}{\ensuremath{\mathpzc{SimplePay}}\xspace}
\newcommand{\payfrompermanent}{\ensuremath{\mathpzc{PayFromPermanent}}\xspace}
\newcommand{\coloredpay}{\ensuremath{\mathpzc{ColoredPay}}\xspace}
\newcommand{\multisigpay}{\ensuremath{\mathpzc{MultisigPay}}\xspace}

\newcommand{\dvd}{\ensuremath{\mathpzc{Divide}}\xspace}
\newcommand{\merge}{\ensuremath{\mathpzc{Merge}}\xspace}

\newcommand{\chal}{\mathcal{C}}
\newcommand{\forgepayment}{\ensuremath{\mathsf{Forge\text{-}Payment}}\xspace}
\newcommand{\forger}{\mathcal{F}}
\newcommand{\sigforge}[1]{\ensuremath{\mathsf{SIG\text{-}FORGE}}_{\forger, #1}(\secpar)}

\usepackage{epigraph}

\begin{document}

\title{Quantum Prudent Contracts with Applications to Bitcoin
\ifdraft{\\(working draft)}
}
\ifnum\anon=0
    \ifnum\sigconf=0
        \author[1]{Or Sattath}
        \affil[1]{Computer Science Department, Ben-Gurion University of the Negev}
    \else
        \author{Or Sattath}
        \affiliation{%
        \institution{Computer Science Department, Ben-Gurion University of the Negev}
        \country{Israel}}
    \fi
\else
    \ifnum\sigconf=0
        \author{}
    \fi
\fi

\ifnum\sigconf=0
    \maketitle
\fi
\begin{abstract}
Smart contracts are cryptographic protocols that are enforced without a judiciary. Smart contracts are used occasionally in Bitcoin and are prevalent in Ethereum. Public quantum money improves upon cash we use today, yet the current constructions do not enable smart contracts. In this work, we define and introduce quantum payment schemes, and show how to implement prudent contracts---a non-trivial subset of the functionality that a network such as Ethereum provides.
Examples discussed include: multi-signature wallets in which funds can be spent by any 2-out-of-3 owners; restricted accounts that can send funds only to designated destinations; and ``colored coins'' that can represent stocks that can be freely traded, and their owner would receive dividends. 
Our approach is not as universal as the one used in Ethereum since we do not reach a consensus regarding the state of a ledger. We call our proposal prudent contracts to reflect this. 

The main building block is either quantum tokens for digital signatures~\cite{BS16,CLLZ21},
semi-quantum tokens for digital signatures~\cite{Shm22}
or one-shot signatures~\cite{AGKZ20}.
The solution has all the benefits of public quantum money: no mining is necessary, and the security model is standard (e.g., it is not susceptible to 51\% attacks, as in Bitcoin). 

Our one-shot signature construction can be used to upgrade the Bitcoin network to a quantum payment scheme.  Notable advantages of this approach are: transactions are locally verifiable and without latency, the throughput is unbounded, and most importantly, it would remove the need for Bitcoin mining. Our approach requires a universal large-scale quantum computer and long-term quantum memory; hence we do not expect it to be implementable in the next few years.
\end{abstract}


\ifnum\sigconf=1
    \maketitle
\fi

\newpage
\tableofcontents
\newpage

\section{Introduction} 
\label{sec:introduction}

\setlength{\epigraphwidth}{0.9\textwidth}
\epigraph{
[Smart contracts are] systems which automatically move digital assets according to arbitrary pre-specified rules. For example, one might have a treasury contract of the form ``A can withdraw up to X currency units per day, B can withdraw up to Y per day, A and B together can withdraw anything, and A can shut off B's ability to withdraw''.
}{Vitalik Buterin~\cite{But14}.}


The advantage of smart contracts, as loosely defined above, over standard legal contracts is that they are enforced by a protocol instead of the judiciary.
Smart contracts have flourished since their conception in the 90s. For example, Bitcoin supports a limited low-level scripting language that can be used to perform multi-sig transactions, atomic cross-chain swaps, and complex two-way payment channels, which culminated in the lightning-network~\cite{NBF+16,PD15}. The raison d'être of Ethereum, another prominent crypto-currency, which is second only to Bitcoin in terms of market cap as of 2022, is to allow \emph{arbitrary} smart contracts by programming the contract in a  Turing-Complete language~\cite{wood14ethereum}. 
An interesting ramification of this idea is a \emph{Decentralized Autonomous Organization} (DAO).\footnote{\emph{The DAO}, an attempt to instantiate this concept, failed spectacularly, see ~\cite{Dup18}.}

Quantum money, invented by Wiesner circa 1969~\cite{Wie83}, was further developed in the past two decades. Several flavors of quantum money were defined and constructed. Some constructions are noise tolerant~\cite{PYJ+12,MVW13,AA17}; quantum coins provide privacy (in the form of untraceability) to their users~\cite{MS10,JLS18,BS20}; there are public quantum money schemes (also called locally verifiable money)~\cite{Aar09,FGH+12,AC13,Zha21,RS22,BS20} in which the bank is not needed in the verification; in semi-quantum money, minting and paying can be done via a classical channel~\cite{RS22,Shm21}; and in quantum lightning schemes, even the bank cannot forge the money~\cite{Zha21,FGH+12}. 

Quantum money, and especially public quantum money, resemble almost all the properties of cash.\footnote{The most significant improvements over cash are that public quantum money provides formal unforgeability guarantees and that (public) quantum money can be transferred using a communication channel, whereas cash requires the payer and the receiver to be at the same location. For a more detailed comparison, see ~\cite{HS21}.} Can we add smart contracts to quantum money? A general solution --- such as the one Ethereum provides --- seems out of reach. Blockchain based solutions rely on reaching consensus about the state of the system, and therefore require ongoing communication.\footnote{See Ref.~\cite{CS20} for a hybrid approach which combines crypto-currencies with quantum money. The approach there  assumes a p2p network that maintains a distributed ledger, whereas this work does not.}
This is in sharp contrast to public quantum money, which is locally verifiable, i.e., transactions involve only the sender and the receiver. Nevertheless, as we show in this work, some limited forms of smart contracts can be achieved. We use the term \emph{prudent} contracts to emphasize that they are not as clever as smart contracts, yet using these contracts is ``exercising sound judgement in practical or financial affairs''.\footnote{The definition is taken from the Oxford English Dictionary.}

\paragraph{Quantum Tokens for Digital Signatures.}
A quantum token for digital signature scheme, or simply a tokenized signature (\ts) scheme, is an extension of digital signatures which can be used to \emph{delegate} the right to sign at most one message. In more detail, the master secret key can be used to generate a single-use quantum signing token. Each quantum signing token has an associated classical public key. A user who holds the quantum signing token can use it to sign an arbitrary message, and the signed message can be verified using the master public key and the public key. The unforgeability guarantee is that an adversary, which receives the master public key and a polynomial number of quantum signing tokens and their public keys, cannot generate two distinct valid signed messages  which pass verification with the same public key. The intuition for unforgeability is that the quantum signing token is measured destructively during the signing algorithm, and therefore cannot be reused to generate another signature. This is formalized in \cref{def:tss}. 

\paragraph{Semi-quantum Tokenized Digital Signatures.}
A semi-quantum tokenized digital signature scheme is the semi-quantum variant of quantum tokenized digital signatures, similar to how semi-quantum money schemes~\cite{RS22,Shm22} are for public-key quantum money schemes. The main difference and advantage over quantum tokenized signatures is that a classical issuer can use the master key to run a classical communication protocol with a quantum user; at the end of which an honest user ends up with a quantum single-use token and a public key associated with it. The unforgeability guarantees are the same as that in the case of tokenized signatures. An adversary who receives the master public key and runs the token generation protocol polynomially many times, cannot generate two distinct signed messages that pass verification with respect to the same public key.
In Ref.~\cite{Shm22} it was shown how semi-quantum tokenized signatures could be constructed based on sub-exponential quantum hardness for LWE and post-quantum indistinguishability obfuscation.

\paragraph{One-shot signatures.}
One-Shot Signatures (\oss) is a cryptographic primitive which was introduced by Amos et al.~\cite{AGKZ20}, which could be viewed as an extension of tokenized signatures and quantum lightning~\cite{Zha21}. Quantum lightning is less relevant for this work.

One-shot signatures could be viewed as tokenized signatures without master keys: the master public key is replaced with a common reference string ($\crs$), and there is no secret key. Instead, anyone could generate a public key and quantum signing key pair. The security guarantee is changed only slightly: the goal of the adversary, who receives the $\crs$ is similar: to produce two distinct signed messages that pass verification with the same public key, and the $\crs$.  Note that this time, there is no need to provide quantum signing tokens to the adversary, since the adversary can produce these signing tokens using the $\crs$. 
The intuition for unforgeability which was mentioned for tokenized signature extends to the generation algorithm of one-shot signatures: the generation algorithm involves a measurement. Therefore repeating the same algorithm would yield a different $(\pk,\qsk)$ pair.

In Ref.~\cite{AGKZ20} it was shown how \oss could be constructed based on a classical oracle, and recently Ref.~\cite{DS22} constructed \oss based on any non-collapsing hash-function. Unfortunately, as far as the author is aware, there are no hash-functions which are proved or even conjectured to be non-collapsing. 
This primitive might appear strange: note that anyone can create a quantum signing key using the $\crs$ (public information). How can that be useful? We discuss the paragraph after the next (\cpageref{par:quantum_money_from_OSS}) how \oss can be used to construct public quantum money, demonstrating this point.

\paragraph{Uncloneable signatures}  For most of the constructions in the article, we either need a tokenized signature, a semi-quantum tokenized signature or an one-shot signature scheme. 
We use the term ``Uncloneable Signature'' to represent all three of them, i.e., an uncloneable signature scheme is either a tokenized signature, a semi-quantum tokenized signature scheme, or an one-shot signature scheme. In our unified framework, a quantum signing key can refer to a tokenized signature token, a semi-quantum tokenized signature token or a one-shot signature quantum signing key. 

\paragraph{Quantum money from uncloneable signatures.}
\label{par:quantum_money_from_OSS}
We will first discuss how \oss imply public quantum money. 
Amos et al.~\cite{AGKZ20} showed how \oss can be used as public quantum money, with the additional novel property that making a transaction does not require any quantum communication channels between the participants. Our next goal is to describe their idea at a high level. Their construction also uses a (post-quantum) digital signature scheme. 

The bank generates a digital signature verification key and (classical) signing key for the digital signature scheme, and publishes the verification key and a \crs.
To mint, the user generates an \oss public key and quantum signing key and sends the \oss public key to the bank. The bank digitally signs the \oss public key.

The quantum money would consist of the \oss quantum signing key and a chain of signatures at any later point. To send the money without using any quantum communication, the sender and the receiver run the following two message protocol.
First, the receiver prepares a new one-shot signature signing key and public key pair and sends its public key to the sender. The sender signs the receiver's public key using the quantum signing key and appends the receiver's public key and the signature to the chain of signatures. Overall, the chain of signatures has the form $\pk_0,\sigma_0,\pk_1,\sigma_1,\pk_2,\sigma_2\ldots,\pk_n,\sigma_n$.
The receiver would check that $\sigma_i$ is indeed a valid signature of the message $\pk_{i+1}$, according to the key $\pk_i$. The receiver would also check that $\pk_0$ is the bank's verification key. Note that the receiver holds the quantum signing key associated with $\pk_n$ and could extend the chain as needed.

Amos et al. also generalize this construction to allow for spending any fraction of a unit of currency by using a primitive called \emph{budget signatures}; we do not use budget signatures in this work. 

We now discuss a similar approach that can be used even with a tokenized signature scheme. In this case, the bank creates a tokenized signature master public key, a master secret key, and the post-quantum digital signature signing key and verification key. The bank uses the master secret key to generate a public key and a quantum signing key for any user asking for such a pair. Note that this is the only part that requires quantum communication between the bank and the user if we use a tokenized signature scheme. We emphasize that a user can ask for many such key pairs. At this point, the quantum signing key do not have any monetary value associated with it. 

Minting is done similarly: the user sends the tokenized signature public key that the user received earlier to the bank. The bank signs the public key using the digital signature signing key, just as in the case of one-shot signatures.
The receiver must have an unused tokenized quantum signing key to receive money. The receiver sends the tokenized public key to the sender. The sender signs the tokenized public key using her quantum signing key. The form of the chain of signatures is the same as before, and verification is done in the same manner.

We can also use a semi-quantum tokenized signature scheme instead of a tokenized signature scheme. The only difference is in the way the user gets the public key and quantum signing key pair. In this case, the bank uses the master secret key to instead run a protocol with classical communication with the user asking for it. The main advantage is that in this case, no quantum communication is required.

\paragraph{Quantum payment schemes.} 
\label{par:payment_schemes}

One of our contributions is the notion of a quantum payment scheme. In existing quantum money schemes, each quantum money state represents some fixed denomination determined by the bank. 
In a quantum payment scheme, the quantum money state is a quantum signing key and is associated with a classical account (which consists of the public key associated with the quantum signing key, as well as other relevant classical data). In some cases, one account can be associated with several quantum signing keys, but that will be explained later in more detail.
Unlike vanilla quantum money, the account can  have an arbitrary non-negative integer balance. 
The account and the associated quantum signing key is created with a $0$ balance. Whether the bank is involved in the creation of an account depends on the choice of instantiation: see \cpageref{par:instantiation} for details. 
The bank can top-up any account the bank chooses, increasing its balance by $1$. The top-up algorithm is classical and generates a (classical) witness that is sent to the user, and certifies that the balance increased by $1$.
The quantum signing key is used only once during payment. 
Unlike quantum money, in which the whole amount is moved from the sender to the receiver, in a payment scheme, payments can reallocate the balance in the account in more advanced ways. A straightforward example is that Alice, that has \$30 in her account, pays \$5 to Bob's account, \$10 to Charlie's account, and the remaining \$15 to another account of hers. 
Interestingly, the payment could be initiated \emph{before} the account has any balance --- a property which is used in some of our applications. Of course the receiving accounts' balance would increase only \emph{after} payments are made to the originating account. 

The payments are made by using the quantum money state (which,  recall, is a quantum signing key) to sign a message in an appropriate format. 
The receiver, in this case, provides the receiving account, and the payment algorithm generates a classical payment witness, which certifies that the receiver's balance increases by the amount sent by the payer. Since the payment witnesses are classical, no quantum channel is needed for making payments, as in the quantum money from one-shot signature construction discussed above. 
To receive a transaction, the receiver needs to have an account and the associated quantum signing key.\footnote{Note that the account can be used to receive many payments, but can be used only once to send payments.}
As discussed in ~\cite{RS22}, the fact that classical communication is used in a transaction has a number of advantages, even if quantum communication is available: (i) unlike classical information, quantum information cannot be resent in case of some communication failure. Such a failure  is especially important here as it could cause money being lost. (ii) A signed version of all the communication between the parties could be recorded. This could be helpful for dispute resolution, internal control, external auditing, etc. Of course, quantum communication cannot be recorded and it might be harder to achieve the goals above.    

The $\balance$ algorithm uses the witnesses generated through the payments to compute the balance of an account. Note that a positive balance of an account is only useful to someone holding the quantum money associated with it. The $\balance$ algorithm does not check that. 

Quantum payment schemes are formally defined in \cref{sec:definitions_of_quantum_payment_scheme}. 

\paragraph{A warm-up quantum payment scheme.}
\label{par:introduction_warm-up_scheme}
In \cref{sec:warm-up_quantum_payment_scheme} we construct our first quantum payment scheme.
This scheme requires a digital signature scheme $\ds$ and an uncloneable signature scheme $\qs$. 
The bank's public key, in this scheme, contains the master public key of the uncloneable signature scheme $\qs$ and the verification key of the digital signature scheme $\ds$.
The bank's secret key is the signing key of $\ds$, and if $\qs$ is a tokenized or a semi-quantum tokenized signature scheme, it also contains the respective master secret key.

To generate an account, the $\qgen$ algorithm of $\qs$ is used: The quantum signing key serves as the quantum money and the public key as the account. Note that if $\qs$ is a tokenized signature scheme, this can be done only by the bank, who would send the quantum signing token and the public key to the user via a quantum channel; If $\qs$ is a semi-quantum tokenized signature scheme, this could be done by running a classical communication protocol with the classical bank with the master secret key and the quantum user; If $\qs$ is an one-shot signature scheme, this could be done by the user directly, without the bank. Recall that at this point, the account has no balance.

To $\topup$ an account, the user sends the account to the bank, and the bank uses the signing key of $\ds$ to sign the account, together with some random string.

To make a payment, the $\simplepay$ algorithm uses the quantum signing key to generate a signature of a message containing all the accounts to pay to. Along with the witness that the paying account has a sufficient balance, this signed message serves as the payment witness for this transaction. 

The $\balance$ algorithm for an $\account$ receives a payment witness list of length that is precisely the balance claimed.
Each such witness can be either a top-up witness, which means the account was topped-up directly by the bank, and in this case, the verification checks that the account and the random number were signed correctly. 
We also check that these (signed) messages are unique so that the same top-up cannot be reclaimed twice: this motivates the inclusion of a random string during top-up, which was mentioned in the penultimate paragraph. 
A witness can be a witness for a payment made from another account. In this case, we check that the signed message is a valid signature with respect to the uncloneable signature scheme, and check recursively that the paying account had a sufficient balance. As was done for top-up witnesses, we need to ensure that the same witness is not reclaimed twice. 

\paragraph{Our main quantum payment scheme construction.}
Our main payment scheme provides more flexibility through several modifications (cf. the warm-up payment scheme above).
\begin{enumerate}
    \item In this construction, an account can be associated with several uncloneable signature public keys. This can be used to allow, e.g., 2-out-of-3 payments, see \cref{sec:backup_and_multiple_signatures} for more details.  This is done by introducing the $\siginterpreter$ program. This is a program that is specified by the user and is part of the account (so, now an account consists not only of the uncloneable signature public keys, but also of $\siginterpreter$). The $\siginterpreter$ receives some messages, some of which could be missing, and outputs a single message, which represents the final decision. So, for example, a $\siginterpreter$ could output the majority of its input (even if some inputs are missing). One example where this is useful is for backup purposes: suppose one of the quantum signing keys in a m-out-of-n account gets lost. In this case, the message associated with the lost quantum signing key would be a missing message.
    Not all $\siginterpreter$ programs are valid. For example, if the account is associated with $2$ public keys, we cannot allow 1-out-of-2 payments, since this would allow a malicious user to double spend, by using the first quantum signing key to transfer the funds to Alice, and the second quantum signing key to transfer the funds to Bob. To overcome this issue, we define the notion of a monotone sig-interpreter (see \cref{def:siginterpreter} for more details), in which this problem cannot occur.  
    \item Recall that in the warm-up scheme, upon payment, the user signed a message containing the list of all accounts to which the funds should be transferred to. In our main scheme, instead, the user signs an $\outputscript$: a program which given $i$ returns the account to which that $i^{\text{th}}$ dollar would be transferred to. The advantage is that we can use a finite program to specify infinitely many outputs. For example, a program could specify that for any $i>10$, the $i^{\text{th}}$ dollar is transferred to $\account_{\text{default}}$. 
    \item Upon payment, the user must specify and sign yet another program, which can block or approve the payment. This program is denoted $\outputverifyscript$. To approve the $i^{\text{th}}$ output, the \emph{receiver} needs to know a verify-script witness, denoted $\vswitness$, such that $\outputverifyscript(i,\vswitness)=1$. This is useful in cases where the sender would like to withhold the payment and reveal the $\vswitness$ at a later step, as part of a more elaborated protocol. An example in which this additional feature is used is Secure Exchange (see \cref{sec:secure_exchange}). In most of our use cases $\outputverifyscript$ always outputs $1$ for every input, and therefore has no meaningful role.
    \item Upon signing, the user can choose to include some classical auxiliary information, denoted $\aux$, in the message to be signed. The $\balance$ algorithm ignores this auxiliary information. $\aux$ can be useful when extending the scheme for other purposes; see such examples in \cref{sec:proof_of_reserves,sec:colored_coins_smart_property_and_tokenized_securities}. 
\end{enumerate}

\paragraph{Applications through extensions.} 
\label{par:applications}
One of the main design goals of our main construction is \emph{flexibility}. As discussed in the previous paragraph, the user has the freedom to choose how many quantum signing keys the account is associated with and the $\siginterpreter$ when creating the account; and to specify the $\outputscript$ and $\outputverifyscript$ programs, as well the auxiliary information $\aux$ upon payment. By using the flexibility that our main construction provides, we show a number of applications, that cash and all the existing quantum money schemes do not provide. 

\begin{enumerate}
    \item \textbf{Divisibility and mergeability.} Banknotes and coins suffer from being harder to use due to issues related to \emph{change}. For example, merchants have to carry enough change to give their customers; this change is also cumbersome for the customers to carry around. 
    
    Cash does not provide any means to \emph{merge} bills: taking two bills with values $v_1$ and $v_2$, and combine them to a single bill with value $v=v_1+v_2$. \footnote{Here, in the context of quantum money, we mean that the number of qubits per bill should remain fixed. Therefore, storing both quantum states would not be considered a legitimate solution.}

    Despite the privacy that cash provides, it is declining as a medium of exchange in many countries~\cite{Jyr04}. Interestingly, its shortcoming concerning divisibility and mergeability (compared to debit cards and other forms of electronic payments) was linked to that decline: Amromin and Chakravorti~\cite{AC07} have found that ``the demand for low denomination notes and coins decreases as debit card usage increases because merchants need to make less change for customer purchases.'' This suggests that resolving this shortcoming may be key to restoring the consumers' privacy.
    
    In \cref{sec:divisibility_and_mergeability} we show that any quantum payment scheme provides divisibility and mergeability. 

    \item \textbf{Permanent Account.}\label{it:permanent_accounts} 
    Bitcoin users can receive multiple payments to the same address (Bitcoin addresses have the same role that accounts have in quantum payment schemes). The sender might save the receiver's address in her wallet's address book. This reduces the need for authenticating the address upon further payments.
    One example is shown in \cref{fig:le_desepere}, where a Bitcoin address is posted in public. We show how to construct \emph{permanent accounts}, i.e., accounts that can accept payments indefinitely. We show a construction for a permanent account using our quantum payment scheme in \cref{sec:permanent_accounts}. 
    \begin{figure}[phtb]
        \centering
        \includegraphics[width=\textwidth]{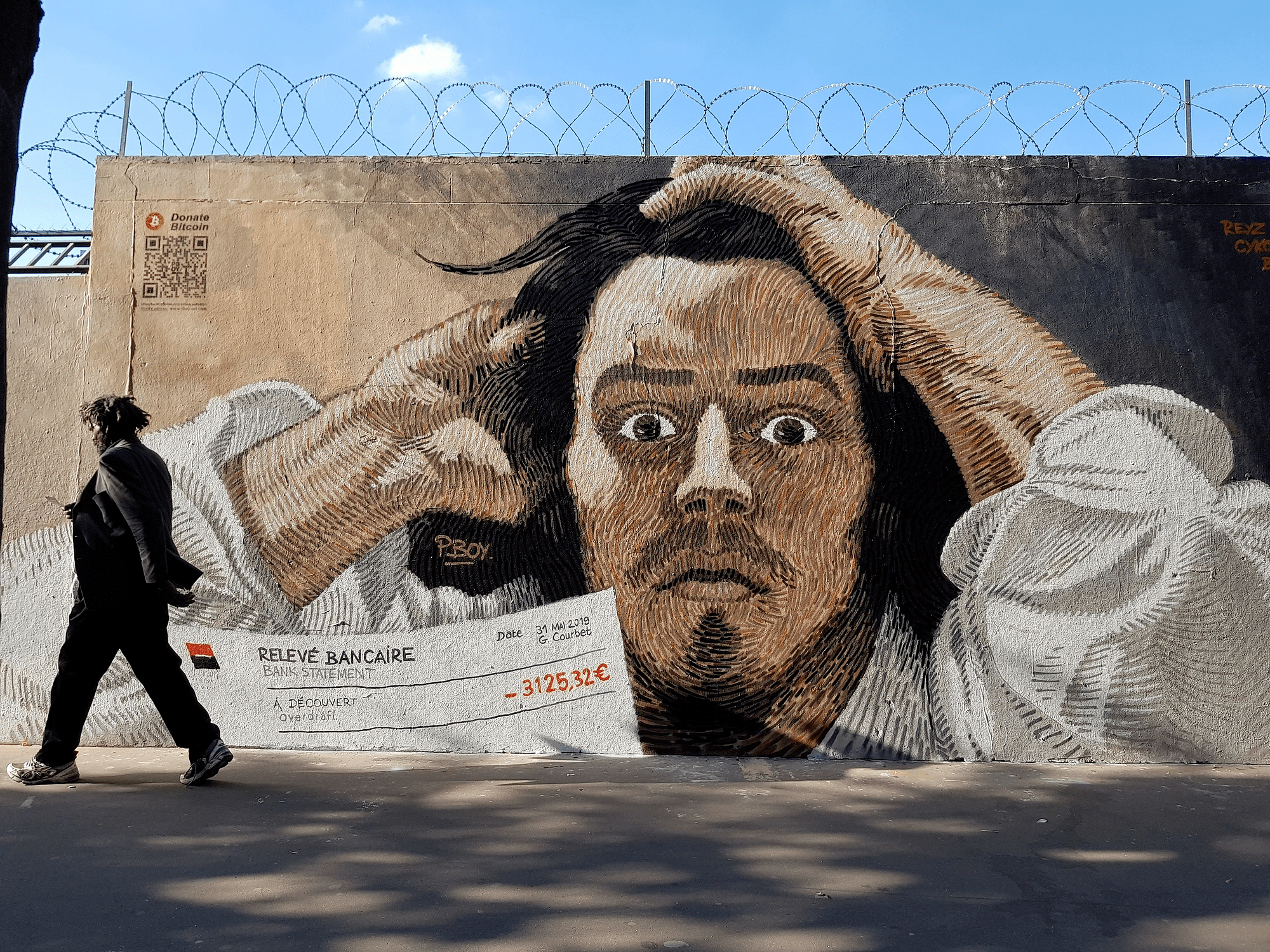}
        \caption{Le Désespéré 2019. Credit: \href{https://www.pboy-art.com/single-post/2019/05/31/fresque-le-d\%C3\%A9sesp\%C3\%A9r\%C3\%A9-2019}{PBoy}. This street-art is inspired by the 
        \href{https://commons.wikimedia.org/wiki/File:Gustave_Courbet_-_Le_D\%C3\%A9sesp\%C3\%A9r\%C3\%A9.JPG}{self-portrait} of Gustave Courbet with (almost) the same name. Bystanders can scan the QR code on the top left of the graffiti with their Bitcoin wallets and send their donations to the artist.     
        }
        \label{fig:le_desepere}
    \end{figure}

    \item \textbf{Secure exchange.} When exchanging currency (e.g., EUR for USD), the first mover is at risk that the second party will cheat and run away with the money. In \cref{sec:secure_exchange} we show a protocol in which the dishonest party can \emph{lock} the other party's funds, but \emph{not} steal it. The same approach can be used to exchange money securely between our quantum payment scheme, and other cryptocurrencies such as Bitcoin, providing the same guarantee. As far as the author is aware, this provides the first example of interoperability between cryptocurrencies and any form of quantum money. 
    
    \item \textbf{Backup and multiple signatures.} There is no way to ``backup'' cash, and losing access to it means that the money is lost. 
    Another important concept for security and internal control is multiple signatures, in which the ability to spend the money requires a subset of the owners. For example, a joint bank account can be set so that two signatures out of the three account owners are needed to sign a check. Cash does not provide this feature. 
    We show a construction that allows backup, in some scenarios, and some forms of multiple signatures, in \cref{sec:backup_and_multiple_signatures}.

    \item \textbf{Restricted account.} A parent might want to deposit funds for a child, which could be spent to some designated list of accounts (for example, money that is intended to be used for college tuition).
    Cash does not provide such a mechanism. Restricted accounts are constructed in \cref{sec:restricted_accounts}.
    
    \item \textbf{Proof of reserves.} A proof of reserves provides a guarantee that the prover possesses a certain amount of money. So far, it was used extensively in Bitcoin~\cite[Section 4.4]{NBF+16} --- e.g., some Bitcoin exchanges prove to their customers that their Bitcoin holdings are at least as high as their liabilities. 
    We show a construction for proof of reserves in \cref{sec:proof_of_reserves}.
    
    \item\textbf{Colored coins.} Cryptocurrencies, such as Bitcoin, can be used to ``color'' a certain coin. Since the history of Bitcoin transactions can be traced, it is fairly easy to transfer ownership of these colored coins and use them for tasks such as unlocking a car or representing company shares that provide voting rights and a mechanism to pay dividends to the shareholders. Our construction for colored coins, smart property and tokenized securities is discussed in \cref{sec:colored_coins_smart_property_and_tokenized_securities}.

\end{enumerate}

 \paragraph{Instantiation.} 
\label{par:instantiation}
Our construction can be instantiated with an uncloneable signature scheme, i.e., either a tokenized signature, semi-quantum tokenized signature, or an one-shot signature scheme.
We now discuss these options in detail:
\begin{itemize}
    \item Tokenized signatures. Coladangelo et al.~\cite{CLLZ21} have shown how to construct tokenized signatures based on post-quantum Indistinguishability Obfuscation ($i\mathcal{O}$), and (the very mild assumptions of) post-quantum one-way functions, and post-quantum digital signatures. Even though some $i\mathcal{O}$ constructions, such as the one which is based on ``well-founded assumption''~\cite{JLS21} are not post-quantum secure, recent constructions are provably secure under post-quantum assumptions, though some of the assumptions are less standard (namely, the circular security of some homomorphic encryption scheme) \cite{BDGM20,GP21}. 
    
    The main disadvantage of this approach compared to the others is that in order to create an account, there needs to be quantum communication between the bank and the user opening the account. 
    
    If the receiver in a transaction does not have any account left, then the sender and the receiver must have quantum communication and can use the same approach which is used for existing public quantum money: The receiver in this case can use the verification algorithm which checks the validity of the quantum signing key---see \cref{rem:verify_token}, and the $\balance$ algorithm to check that account's balance. This way for handling the case where the receiver does not have an account also holds for semi-quantum tokenized signatures described next.
    \item Semi-quantum tokenized signatures. Recall that in the case of tokenized signatures, the bank needs to create the account, which contains a quantum signing key, and send the quantum signing key to the user via a quantum channel. With semi-quantum tokenized signatures, there is no need for a quantum channel: the (now classical) bank and the user could run the (interactive) token generation protocol instead.
    
    Recently, semi-quantum tokenized signatures were formally defined and  constructed based on sub-exponential hardness of $LWE$ against \qpt adversaries, and the existence of post-quantum $i\mathcal{O}$, see,~\cite{Shm22}.
    \item One-shot signatures. Here, users do not need to interact with the bank other than for $\topup$s since an account can be created without the bank's assistance.
    
    The only construction known, by Amos et al.~\cite{AGKZ20} is relative to a classical oracle.
\end{itemize}

\paragraph{A transition from Bitcoin mining.}
In \cref{sec:transition_from_bitcoin_minig}, we outline an upgrade path for Bitcoin and other similar cryptocurrencies. The approach we use is similar to the one used in Ref.~\cite{CS20}. 
The prudent contracts in this work provide most of the functionality used in Bitcoin.
The main advantages of our approach compared to the current way the Bitcoin network operates are the following. Our approach does not require mining (and therefore does not consume enormous amounts of energy and does not assume an honest majority). Transactions are locally verifiable, incur no fees, and have low latency. Our solution crucially relies on one-shot signatures. Since the upgrade requires universal quantum computers and long-term quantum memories, it would not be practical in the next few years to the very least.

\paragraph{Organization.} 
\label{par:organization}
The outline of the rest of the paper is as follows. In \cref{sec:definitions} we introduce some notations, and preliminary definitions, most importantly, tokens for digital signatures, semi-quantum tokenized digital signatures and one-shot signatures. In \cref{sec:definitions_of_quantum_payment_scheme} we define quantum payment schemes, 
and in \cref{sec:why_existing_techniques_not_applicable} we discuss why existing scripting languages, or other domain specific languages are not easy to adapt to our setting. In \cref{sec:warm-up_quantum_payment_scheme} we construct a warm-up quantum payment scheme, and prove its security. In \cref{sec:divisibility_and_mergeability} we show how any quantum payment scheme, such as the warm-up scheme, provides divisibility and mergeability. In \cref{sec:construction_payment_scheme} our main quantum payment scheme is introduced. The security of the scheme is proved, by reusing most of the security proof of the warm-up scheme. The added features of our main construction are used in several applications, which are discussed in \cref{sec:applications}. A proposed upgrade to the Bitcoin network which removes the need for Bitcoin mining is presented in \cref{sec:transition_from_bitcoin_minig}.

\section{Notations and preliminaries} 
\label{sec:definitions}

We recommend the following introductory textbooks: for classical cryptography see ~\cite{KL14,Gol01,Gol04}, for cryptocurrencies see~\cite{NBF+16}, for quantum computing see~\cite{NC11}. See also the following survey for quantum cryptography~\cite{BS16b}. 

We mostly follow the notation conventions by Amos et al.~\cite{AGKZ20}.
We use calligraphic fonts with a capital letter to represent quantum algorithms (e.g., \qalg), lowercase calligraphic font for quantum states (e.g., \qsk for a quantum signing key), capitalized sans-serif for classical algorithms (e.g., $\mathsf{Alg}$) and lowercase classical variables (e.g., $\account$) and also standard math font for single-letter variables (e.g., $n$ and $i$). For cryptographic schemes, we use all capitals (e.g., $\ds$ for a digital signature). To avoid ambiguity, we use the scheme followed by the name of the algorithm (e.g., $\ds.\verify(\cdot)$).

For a finite set $S$, we use the shorthand $x \gets S$ for the random process where $x$ is sampled uniformly at random from $S$. Similarly, $x \gets \alg(\cdot)$ is the process in which $x$ is sampled from the distribution $\alg(\cdot)$ as specified by the algorithm. 
A function $\epsilon$ is negligible if $\epsilon(n)\in o(n^{-c})$ for all $c\in \NN$. 

While stating a new algorithm we often compare it to a similar one, and use strike-through to represent parts from the previous algorithm which are removed from the current algorithm (e.g., \sout{some removed part}), and curly underlines to represent additional parts that were not present in the previous algorithm (e.g., \uwave{an additional part}).

Next, we will define a new algorithm called $\verify^2$ that is based on the algorithm $\verify$ that will be used in the unforgeability definitions of all the three uncloneable signature schemes.

\begin{definition}\label{def:verify-2}
Given an algorithm $\verify$ which takes as an input $\mpk,\pk,m,\sigma$ we define $\verify^2(\mpk,\pk,m_0,m_1,\sigma_0,\sigma_1)$ that returns $1$ if and only if 
\begin{enumerate*}
    \item $m_0\neq m_1$, \label{it:check_m0m1}\\
    \item $\verify(\mpk,\pk,m_0,\sigma_0)=1$, and 
    \item $\verify(\mpk,\pk,m_1,\sigma_1)=1$.
\end{enumerate*}
\label{def:verify_2}
\end{definition}
\begin{remark}
    If we replace \cref{it:check_m0m1} with $(m_0,\sigma_0)\neq (m_1,\sigma_1)$ in the definition of $\verify^2$ above and use it instead in the unforgeability definitions, then we get a stricter notion of unforgeability called strong unforgeability. The difference between standard unforgeability (i.e., without the modification in \cref{def:verify_2}) and strong unforgeability has been discussed in detail in~\cite{BSW06,ADR02}. 
    One-shot signatures as defined in~\cite{AGKZ20}
    satisfy strong unforgeability. 
    However in this work, we only use standard unforgeability.
    \label{rem:standard_vs_weak_unforgeability} 
\end{remark}

\begin{definition}[Tokens for digital signatures]
A tokenized signature scheme consists of 4 algorithms, \setup (\ppt), \tokengen\ (\qpt), \qsign\ (\qpt), and \verify\ (deterministic polynomial time) with the following syntax:
\begin{enumerate}
\item $(\mpk,\msk) \gets \setup(\secparam)$: takes a security parameter and outputs a classical master public key \mpk, and a classical master secret \msk.
\item $(\pk,\qsk) \gets \tokengen(\msk)$: takes the master secret key $\msk$ and outputs a classical public key $\pk$, and a quantum signing key $\qsk$. We emphasize that if $\tokengen(\msk)$ is called $\ell$ times, it may (and for reasons related to unforgeability, which will become apparent soon, should, with overwhelming probability) output different states $\qsk_1,\ldots,\qsk_{\ell}$ and different public keys $\pk_1,\ldots,\pk_\ell$.
\item $\sigma \gets \qsign(\qsk,m)$: takes a quantum signing key $\qsk$, and a classical message $m$, and outputs a classical signature $\sigma$.
\item $b \gets \verify(\mpk,\pk,m,\sigma)$: receives a classical master public key, a classical public key, a classical message $m$ and an alleged classical signature. It outputs a bit $b$, with $b=1$ meaning \textbf{valid} and $b=0$ meaning \textbf{invalid}.
\end{enumerate}
\paragraph{Correctness.} If $(\mpk,\msk)\gets \setup(\secparam)$ and $(\pk,\qsk)\gets \tokengen(\msk)$ then for any message $m\in\{0,1\}^*$, $\verify(\mpk,\pk,m,\qsign(\qsk,m))=1$ with overwhelming probability. 
\paragraph{Unforgeability.} For any $\qpt$ adversary $\qadv$, there is a negligible function $\epsilon$ such that for all $\secpar$:
\begin{equation}
\Pr\left[
\begin{array}{c}
    (\mpk,\msk)\gets \setup(\secparam)\\
    (\pk,m_0,m_1,\sigma_0,\sigma_1) \gets \qadv^{\tokengen(\msk)}(\mpk)
\end{array}
:
\begin{array}{c}
    \verify^2(\mpk,\pk,m_0,m_1,\sigma_0,\sigma_1)=1
\end{array}
 \right]\leq \epsilon(\secpar)
\label{eq:ts_unforgeability}
\end{equation}
\label{def:tss}
\end{definition} 

 \begin{remark}
\label{rem:syntax-tokenized}
    The definitions of tokenized signatures and semi-quantum tokenized signatures (\cref{def:tss,def:semi-quantum-tokenized-signatures} respectively) are slightly different, both it term of syntax and unforgeability than the ones used in prior works~\cite{BS16,CLLZ21,Shm22}. We adapted the definition so that we could use a unified framework that uses very similar unforgeability definitions for tokenized and semi-quantum tokenized signatures and one-shot signatures. The difference between the original definition and the one here, as well as the applicability of the existing constructions to the definition above, is discussed in detail in \cref{sec:tokenized_signatures_unforgeability}.
\end{remark}

\begin{remark}
In prior works, the tokenized signature scheme had a fifth algorithm, which verifies the validity of the quantum signing key. This algorithm, as well as the unforgeability properties are mostly irrelevant for this work, and therefore are omitted.
\label{rem:verify_token}
\end{remark}
 
Coladangelo et al. showed that there exists a tokenized signature mini-scheme (see \cref{def:tss_mini_scheme} in \cref{sec:tokenized_signatures_unforgeability}) under some cryptographic assumptions, see \cref{thm:tokenized_signatures_mini-scheme_exist}. Based on this, we show the following result.

\begin{theorem}\label{thm:tokenized_signatures_exist}
Assuming post-quantum indistinguishability obfuscation and one-way function 
there exists an unforgeable tokenized signature (see \cref{def:tss}). 
\end{theorem}

The proof of this theorem is given in \cref{sec:tokenized_signatures_unforgeability} on \cpageref{pf:thm:tokenized_signatures_exist}.

    In semi-quantum money~\cite{RS22,Shm21},  minting is not an algorithm but rather an interactive protocol between a \emph{classical} bank (which has the master secret key) and the user, which has a quantum computer.
    
    One can naturally define semi-quantum tokenized signature schemes in a similar manner, which were formalized and constructed in a recent work, see~\cite{Shm22}.
    The main advantage of using semi-quantum tokenized signatures over standard tokenized signatures is that no quantum communication infrastructure would be needed. 
\begin{definition}[{Semi-quantum tokenized signature}]
 A semi-quantum tokenized signature scheme consists of three algorithms ($\setup$, $\qsign$, $\verify$) with the same syntax as tokenized signatures. The $\qgen$ algorithm is replaced with the $\qgen$ protocol with classical communication with the following syntax:
 $(\pk,\qsk)\gets \langle\qgen.\sen(\msk),\qgen.\rec\rangle_{(\out_\sen,\out_\rec)}$: is a classical-communication protocol between two parties, a $\ppt$ sender $\qgen.\sen$ which gets the master secret key $\msk$ as input and a $\qpt$ receiver $\qgen.\rec$ (without any input). At the end of interaction, the sender outputs a classical public key $\pk$ and the receiver outputs a quantum signing key $\qsk$. As in the case of tokenized signature schemes, if this token generation protocol is called $\ell$ times, it may (and for reasons related to unforgeability, which will become apparent soon, should, with overwhelming probability) output different states $\qsk_1,\ldots,\qsk_{\ell}$ and different public keys $\pk_1,\ldots,\pk_\ell$.
  
 \paragraph{Correctness.} If $(\mpk,\msk)\gets \setup(\secparam)$, and $(\pk, \qsk ) \gets \langle\qgen.\sen(\msk),\qgen.\rec\rangle_{(\out_\sen,\out_\rec)}$, then $\verify(\mpk, \pk, m, \qsign(\qsk , m)) = 1$ for any message $m\in \{0,1\}^*$ with overwhelming probability.
 
 \paragraph{Unforgeability.\protect\footnote{See \cref{rem:standard_vs_weak_unforgeability}}} For any $\qpt$ algorithm $\qadv$, there is a negligible function $\epsilon$ such that for all $\secpar$,
\begin{equation}
\Pr\left[ 
\begin{array}{c}
    (\mpk,\msk)\gets \setup(\secparam)\\
    (\pk,m_0,m_1,\sigma_0,\sigma_1 ) \gets \qadv^{\langle\qgen.\sen(\msk),\cdot\rangle}(\mpk)
\end{array}
:
\begin{array}{c}
   \verify^2(\mpk,\pk,m_0,m_1,\sigma_0,\sigma_1)=1
\end{array}
\right]
\leq \epsilon(\secpar).
\label{eq:sqts_unforgeability}
\end{equation}
\label{def:semi-quantum-tokenized-signatures}
 In the equation above (\cref{eq:sqts_unforgeability}), $\qadv^{\langle\qgen.\sen(\msk),\cdot\rangle}$ represents that $\qadv$ has oracle access to the $\qgen$ protocol with $\sen(\msk)$. 
\end{definition}

Recently, Shmueli proved the existence of a semi-quantum tokenized signatures mini-scheme (see \cref{def:semi-quantum-tss_mini_scheme}) under cryptographic assumptions, see \cref{thm:semi_quantum_tokenized_signatures_mini-scheme_exist}, based on which we show the following result.
\begin{theorem}\label{thm:semi-quantum_tokenized_signatures_exist}
Assume that Decisional LWE has sub-exponential quantum indistinguishability 
and that indistinguishability obfuscation for classical circuits exists with security against quantum polynomial
time distinguishers. Then, there is a semi-quantum tokenized signature scheme.
\end{theorem}
The proof of this theorem is given in \cref{sec:tokenized_signatures_unforgeability} on \cpageref{pf:thm:semi-quantum_tokenized_signatures_exist}.

\begin{definition}[One-shot signatures (Adapted from~\cite{AGKZ20})]
A one-shot signatures is a tuple of algorithms $(\qgen,\qsign,\verify)$ with the following syntax:
\begin{enumerate}
    \item $(\pk,\qsk) \gets \qgen(\crs)$: takes a common reference string and outputs a classical public key $\pk$ and a quantum signing key $\qsk$.
    \item $\sigma \gets \qsign(\qsk,m)$: takes a secret key $\qsk$ and a message $m$ and outputs a classical signature $\sigma$.
    \item $b \gets \verify(\crs, \pk, m, \sigma)$: takes a common reference string $\crs$, a public key $\pk$, a message $m$ and a signature $\sigma$ and outputs a bit $b$. 
\end{enumerate}
\paragraph{Correctness.} If $(\pk, \qsk ) \gets \qgen(\crs)$ then $\verify(\crs, \pk, m, \qsign(\qsk , m)) = 1$ for any message $m\in \{0,1\}^*$ with overwhelming probability.

\paragraph{Unforgeability.\protect\footnote{See \cref{rem:standard_vs_weak_unforgeability}}} For any quantum polynomial time algorithm $\qadv$, there is a negligible function $\epsilon$ such that for all $\secpar$,
\begin{equation}
\Pr\left[ 
\begin{array}{c}
    \crs \gets \{0, 1\}^\secpar \\
    (\pk,m_0,m_1,\sigma_0,\sigma_1 ) \gets \qadv(\crs)
\end{array}
:
\begin{array}{c}
    \verify^2(\mpk,\pk,m_0,m_1,\sigma_0,\sigma_1)=1
\end{array}
\right]
\leq \epsilon(\secpar).
\label{eq:oss_unforgeability}
\end{equation}
\label{def:one-shot-signatures}
\end{definition}

Amos et al. proved:
\begin{theorem}[\cite{AGKZ20}]
There exists a classical oracle relative to which one-shot signatures exist.
\label{thm:one-shot-signatures-exist}
\end{theorem}

Our constructions can use a tokenized, semi-quantum tokenized, or a one-shot signature scheme. As mentioned in the introduction, we use \emph{uncloneable signature scheme} as a term that can be used to refer to any of the three schemes. In order to have compatibility between the syntax of tokenized signatures, semi-quantum tokenized signatures and one-shot signatures, we also define a $\setup(\secparam)$ algorithm for one-shot signatures, which samples a $\crs$ uniformly at random from $\{0,1\}^\secpar$. Also for compatibility reasons, we use the notation $\mpk$ as the output of $\setup$, instead of a $\crs$, so that we have a unified framework for tokenized and semi-quantum tokenized signatures and one-shot signatures. Note that in one-shot signatures, there is no master secret key $\msk$. We use the convention that for semi-quantum tokenized signatures, $(\pk,\qsk)\gets \qgen(\msk)$ represents $(\pk,\qsk)\gets \langle\qgen.\sen(\msk),\qgen.\rec\rangle_{(\out_\sen,\out_\rec)}$, i.e., the generation of the public key and the quantum signing key via the classical communication protocol, $\qgen$. Similarly, in the context of semi-quantum tokenized signatures, we say that an adversary $\qadv$ has oracle access to $\qgen$, i.e., $\qadv^{\qgen(\msk)}$ to represent $\qadv^{\langle\qgen.\sen(\msk),\cdot\rangle}$.

\begin{definition}[{Digital signature scheme \cite[Definition 12.1]{KL14}}]\label{def:digital_signature}
    A digital signature scheme consists of two \ppt algorithms $\keygen$, $\sign$ and a deterministic polynomial time algorithm $\verify$ such that:
    \begin{enumerate}
        \item The key-generation algorithm $\keygen$ takes as input a security parameter $\secparam$ and outputs a pair of signing and verification keys $(\vk, \sk)$. We assume that $\vk$ and $\sk$ have lengths of at least $\secpar$ and that $\secpar$ can be determined from either.
        \item The signing algorithm $\sign$ takes as input a private key $\sk$ and a message $m$. It outputs a signature $\sigma \gets \sign_{\sk}(m)$.
        \item The deterministic verification algorithm $\verify$ takes as input a verification key $\vk$, a message $m$ and a signature $\sigma$. It outputs a bit $b \gets \verify_{\vk}(m,\sigma)$, with $b=1$ meaning \textbf{valid} and $b=0$ meaning \textbf{invalid}.
    \end{enumerate}
    \emph{Completeness:} except with negligible probability over $(\vk, \sk)$ output by $\keygen(\secparam)$, it holds that $\verify_{\vk}(m, \sign_{\sk}(m))=1$ for every message $m$.
\end{definition}

The bank is the one that needs to use the digital signature scheme. Our adversaries might be quantum, yet the bank is classical. Hence, we use post-quantum CMA security rather than qCMA security (see \cite{BZ13}): the adversary cannot make superposition signing queries to the bank since the bank is classical. 

\begin{definition}[{PQ-EU-CMA digital signature scheme, adapted from \cite[Definition 12.2]{KL14}}]\label{def:unforgeability_of_digital_signature}
    A digital signature scheme $\Pi$ is \emph{Post-Quantum Existentially Unforgeable under an adaptive Chosen Message Attack} (PQ-EU-CMA) if for every \qpt forger $\forger$, there exists a negligible function $\negl$ such that:
    \begin{equation}
        \Pr[\sigforge{\Pi} = 1] \leq \negl\;.
    \end{equation}
    The signature-forge experiment $\sigforge{\Pi}$ is defined as follows:
    \begin{enumerate}
        \item $\keygen$ is run to generate to obtain keys $(\vk, \sk)$.
        \item Forger $\forger$ is given $\vk$ and access to a (classical) signing oracle $\sign_{\sk}(\cdot)$. We emphasize that the forger cannot query the signing oracle in superposition. The forger than outputs $(m, \sigma)$. Let $Q$ denote the set of all classical queries that $\forger$ made to the signing oracle.
        \item $\forger$ succeeds iff $\verify_{\vk}(m, \sigma)=1$ and $m \notin Q$. In this case the output of the experiment is defined to be 1 (and otherwise 0).
    \end{enumerate}
\end{definition}

\section{Definitions of a quantum payment scheme}
\label{sec:definitions_of_quantum_payment_scheme}
A quantum payment scheme consists of five algorithms, $\setup$ and $\topup$ which are \ppt, $\simplepay$ and  $\createsimpleaccount$ which are \qpt, and $\balance$ which is deterministic and polynomial time, with the following syntax:
\begin{enumerate}
    \item $(\bpk,\bsk) \gets \setup(\secparam)$.
    \item $(\ket{\$},\account) \gets \createsimpleaccount(\bsk)$.%
    \footnote{The main advantage of using a one-shot signature is that in this case, $\createsimpleaccount$ only uses the bank's public key $\bpk$, and therefore can be done without the bank's involvement.}
    \item $\pwit \gets \topup_{\bsk}(\account)$: gets the bank's secret key $\bsk$ and an account, and creates a payment witness. This payment witness serves as a proof that the $\balance$ of $\account$ increased by 1. 
    
    \item $(\pwit_1,\dots,\pwit_r) \gets \simplepay(\bpk,\ket{\$},\account, \witness, (\account_1,\ldots,\account_r))$: generates $r$ payment witnesses. If payments are made from an account which has balance $r$, than $\pwit_i$ can be used to certify that the balance of $\account_i$ increased by $1$.
    \item $j \gets \balance(\bpk, \account, \witness): $ gets the $\bpk$, an $\account$, and a $\witness$ and outputs a non-negative integer representing the balance of this account.
\end{enumerate}

The completeness property requires that the following conditions hold. 
\begin{enumerate}
    \item \label{it:balance_respects_top-up} \textbf{Balance respects top-up.} Informally, if $\balance(\bpk,\account,\witness)$ is $j$, then 
    
    $\topup_{\bsk}(\account)$ generates a witness which increments the $\account$ balance by 1. 
    Formally, the bank's public and secret key-pair is generated using $\setup(\secparam)$. Then, an $\account$ is created using $\createsimpleaccount(\bsk)$. Through arbitrary series of simple payments and top-ups, we have that 
 \[\balance(\bpk,\account,\witness)=j.\] 
Now, a top-up is made to $\account$: $\pwit \gets \topup_{\bsk}(\account)$.
Let $\witness'$ be the concatenation of the $\witness$ and $\pwit$. 
We require that under these circumstances, for every $\secpar$,
\[ \Pr[\balance(\bpk,\account,\witness')=j+1]=1.\]
    \item \label{it:balance_respects_payments} \textbf{Balance respects payments.} Informally, suppose account $A$ had balance $j$ according to $\witness$, and that it receives another payment. Recall that $\simplepay$ generates a payment witness $\pwit$. We require that $A$ would have a balance $j+1$ with the witness $\witness'=\witness||\pwit$. 

    Formally, the bank's public and secret key-pair is generated using $\setup(\secparam)$. Then two simple accounts are created: \[ (\ket{\$_i},\account_i) \gets \createsimpleaccount(\bsk), \quad \text{for $i\in\{0,1\}$}.\]
    Through arbitrary series of simple payments and top-ups, we have for $i\in\{0,1\}$ that 
    \[\balance(\bpk,\account_i,\witness_i)=j_i.\] 
    Now, a simple payment is made from $\account_0$:
    \[(\pwit_1,\ldots,\pwit_{j_0})\gets \simplepay(\bpk, \ket{\$_0},\account, \witness'_0,(\account'_1,\ldots,\account'_{j_0})).\]
    Furthermore, suppose $\account_1$ is the recipient in $t$ of these payments (formally, 
    
    $|\{i | \account'_i=\account_1 \}|=t$). 
    Let $\pwit$ be the concatenation of these payment witnesses (formally, $\pwit$ is the concatenation of the elements $\{\pwit_i | \account'_i=\account_1 \}$ in an arbitrary order). 
    We require that under these circumstances, for every $\secpar$,
    \[ \Pr[\balance(\bpk,\account_1,\witness_1||\pwit)=j+t]=1.\]

    \item \label{it:additive_polynomial_blow_up} \textbf{Additive polynomial blow-up.} Under the setting described in the previous item, the size of the payment witnesses should increase by at most a $\poly$ additive factor. This guarantees that the length of the entire witness remains polynomial in $\lambda$ and linear in the number of times that the money "switched hands." 
    For example, this rules out payment schemes in which the size of the witness doubles after every payment.
    Formally, we require that there exists some polynomial $p$ such that in the setting specified in the previous item, for every $j \in [j_0]$, $|\pwit_j|\leq |\witness_0|+p(\secpar+j_0)$. 

\end{enumerate}

We now turn to discuss unforgeability. Intuitively, a \qpt adversary who uses $n$ calls to \topup should not be able to produce strictly more than $\$n$, except with a negligible probability. In the security game, the adversary can ask for as many fresh accounts, and his goal is to ``load'' these accounts with a total balance of $n+1$ using only $n$ calls to $\topup$. This is formalized in the security game $\forgepayment_{\qadv,\Pi}(\secpar)$.

\begin{game}
\caption{Forgeability game $\forgepayment_{\qadv,\Pi}(\secpar) $.
}\label{game:unforgeability} 
\begin{algorithmic}[1]
    \State The adversary $\qadv$ sends $m,\ell$ to the challenger $\chal$.
    \State $\chal$ prepares $(\bpk,\bsk)\gets \setup(\secparam)$.
    \State \label{line:target_accounts}$\chal$ creates $m$ target accounts:   $(\ket{\$_i},\account_i) \gets \createsimpleaccount(\bsk)$.
    \State $\chal$ sends $\bpk, \{\account_i\}_{i=1}^m$, to $\qadv$. 
    \State \label{line:adversary_prepares_witnesses}$\qadv$, while having oracle access to $\topup_{\bsk}(\cdot)$ and $\createsimpleaccount(\bsk)$, prepares $m$ new witnesses $(\witness_i)_{i=1}^m$, and sends them to $\chal$. Denote by $n$ the total number of queries $\qadv$ makes to the $\topup$ oracle. 
    
    \State \label{line:calculating_balances}$\chal$ computes $b_i \gets \balance(\bpk,\account_i,\witness_i)$ for all $i \in \{1,\ldots,m\}$.
    
    \State The outcome of the game is $1$ (and we say $\qadv$ wins) if $\sum_{i=1}^m b_i > n$.
\end{algorithmic}
\end{game}

\begin{definition}
A quantum payment scheme $\Pi$ is unforgeable if for any \qpt in $\secpar$ adversary $\qadv$, $\Pr(\forgepayment_{\qadv,\Pi}(\secpar) =1) \leq \negl$.
\label{def:unforgeable_payment}    
\end{definition}

As implicitly implied by the names $\createsimpleaccount$ and $\simplepay$, the payment scheme has some freedom to use more advanced forms of accounts and payments, as we will see in our construction.

\subsection{Why existing techniques are not applicable?}
\label{sec:why_existing_techniques_not_applicable}
Existing cryptocurrencies use various Domain Specific Language (DSLs) scripts to express their transactions --- see~\cite{STM16}.
One may wonder: why can't we reuse an existing scripting language, such as Solidity --- a high-level language used in Ethereum? We identify two barriers that make the existing cryptocurrencies scripting languages unfit for our purposes: 
\begin{enumerate}
    \item \label{it:existance_of_DLT} The existing cryptocurrency scripting languages rely on some form of a distributed ledger (usually, a blockchain). Participants in the network communicate and eventually agree on the state of some form of ledger. In contrast, in this work, the system function in a way very similar to cash, where the only parties that take place in the transaction are the sender and the receiver, and they do not need to communicate with others. 
    \item \label{it:reliance_on_multiple_signing} Existing cryptocurrencies extensively rely on digital signatures. It is assumed that the owner of the signing key can sign as many signatures as she wishes. In the context of \qs, one has to consider that it is impossible to sign more than one message. We note that there is a similarity between \qs and one-time signatures~\cite[Section 6.4.1]{Gol04}, which might be confusing: in \qs, the signer cannot sign more than one message with the same quantum signing key, while in onetime signatures, even though it is possible to sign multiple messages using the same classical signing key, it is insecure to do so, and may lead to forgeries.
\end{enumerate}
We now discuss these two items. \cref{it:reliance_on_multiple_signing} is less crucial: the inability to sign multiple messages using the same signing key would imply that some \emph{use cases} may not be implementable; it does not seem to create serious security issues. On the other hand, \cref{it:existance_of_DLT} does create severe security risks --- namely, double-spending. 

To demonstrate \cref{it:existance_of_DLT}, recall that every Bitcoin transaction contains a short script that could be viewed as a lock. The spending transaction provides some data (a witness) that ``unlocks'' it~\cite{NBF+16}. This locking script is called scriptPubKey, and the unlocking script is called scriptSig --- see, e.g., Ref.~\cite{Ant17} for the low-level mechanics. One notable example is called a 1-of-2 multi-sig transaction. In this case, the spending transaction that effectively allows any one of two signing keys to spend these outputs. Concretely, suppose Alice generates a private and public key-pair, $\sk_A,\pk_A$ and similarly Bob creates $\sk_B,\, \pk_B$. Charlie can create a transaction which allows either Alice or Bob to spend it: to unlock the Bitcoins that Charlie sends, a valid signature associated with $\pk_A$ or $\pk_B$ is needed. Suppose Alice and Bob spend the transaction at the same time by signing with their respective secret key, and sending the transaction to the miners? A miner will have to choose at most one transaction from these two to include in the block, and only one of the two transactions could be in the longest blockchain (due to verification rule which do not permit double-spending the same input). Since we use a \qs instead of a distributed ledger, we need a  different mechanism to prevent double-spending a multi-sig transaction that does not rely on the distributed ledger. 

Next, we discuss \cref{it:reliance_on_multiple_signing}. We first provide a few examples where signing multiple messages using the same private key is useful in a network such as Bitcoin. Recall that Bitcoin transactions have to attach a fee to be included in a block. Suppose Alice tries to send some bitcoins to Bob, but the fee she attaches is too low, and it does not get included in the next few blocks. One way for Alice to remedy this is to sign another transaction with a higher fee which would get a higher priority and increase the likelihood of her transaction getting confirmed\footnote{For more details, see \href{https://github.com/bitcoin/bips/blob/master/bip-0125.mediawiki}{BIP-125}. Another, slightly more complicated method, known as child-pays-for-parent, also \href{https://bitcoin.org/en/release/v0.13.0\#mining-transaction-selection-child-pays-for-parent}{exists}. Child-pays-for-parent does not require signing using the same secret key.}%
. 
Consider another example. Bitcoin addresses are sometimes used for an indefinite time. One use case where this is useful is permanent addresses --- see \cref{it:permanent_accounts} on \cpageref{it:permanent_accounts}. If transactions were limited to signing a single transaction using the same private key, the artist in this example would have needed to replace the QR code after each use, which would be impractical.

A use-case in which Bitcoin crucially relies on signing multiple messages with the same signing keys is the lightning-network~\cite{PD15}.

\section{A warm-up quantum payment scheme} 
\label{sec:warm-up_quantum_payment_scheme}
In this section, we provide a warm-up instantiation of a quantum payment scheme. The quantum payment offers little functionality beyond simple payments.
In \cref{sec:construction_payment_scheme} we will introduce several more features. 
With this warm-up construction, we will add the divisibility and mergeability property --- the other prudent contracts require some of the advanced features of the full scheme, which will be presented in \cref{sec:construction_payment_scheme}.
\begin{figure}[phtp]
\noindent\textbf{Assumes:} Digital Signature $\ds$, and $\qs$ is an uncloneable signature scheme.\\
\noindent$\setup(\secparam)$
\begin{compactenum}
\item $(\sk,\vk) \gets \ds.\keygen(\secparam)$.
\item $(\mpk,\msk) \gets \qs.\setup(\secparam)$. 
\item $\bpk\equiv (\mpk,\vk)$, $\bsk \equiv (\msk,\sk)$.
\item Output $(\bpk,\bsk)$.
\end{compactenum}

\noindent $\createsimpleaccount(\bsk)$ \hfill // If one-shot signatures are used, $\bpk$ is used instead of $\bsk$.
\begin{compactenum}
\item Interpret $\bsk$ as $(\msk,\sk)$. 
\item $(\pk, \qsk) \gets \qs.\qgen(\msk)$.  \hfill // If $\qs$ is a one-shot signature scheme, $\msk$ should be replaced with $\mpk$.
\item Output $ (\ket{\$}\equiv \qsk,\ \account \equiv \pk)$. 
\end{compactenum}

\noindent $\topup_{\bsk}(\account)$
\begin{compactenum}
\item Interpret $\bsk$ as $(\msk,\sk)$.
\item Let $\rand\gets_R\{0,1\}^\lambda$.
\item Set $m\equiv (\account,\rand)$.
\item $\sigma \gets \ds.\sign_{\sk}(m)$.
\item Output $\pwit \equiv (\rand,\sigma) $.
\end{compactenum}

\noindent $\simplepay(\bpk, \ket{\$},\account, \witness,(\account_1,\ldots,\account_r) )$ \\
Assumes that $\balance(\bpk,\account,\witness)=r$.
\begin{compactenum}
\item Set $m \equiv (\account_1,\ldots,\account_r)$.
\item $\sigma \gets \qs.\qsign(\ket{\$},m)$.
\item For every $i\in [r]$:
\begin{compactenum}
    \item \label{it:simple_pay_pwit} Set $\pwit_i \equiv (\account, m,\sigma,i,\witness)$
\end{compactenum}
\item Output $(\pwit_1,\ldots,\pwit_r) $.
\end{compactenum}

\noindent $\balance(\bpk, \account,\witness)$
\begin{compactenum}
\item Interpret $\bpk$ as $(\mpk,\vk)$.
\item Interpret $\witness$ as $(\witness_1,\ldots, \witness_v)$. \hfill // $v$ is the claimed balance. 
\item For every $i\in\{1,\ldots,v\}$:
\begin{compactenum}
    \item If $\witness_i.size=2$: \hfill // A top-up witness.  \label{it:top-up-witness}
    \begin{compactenum}
        \item Interpret $\witness_i$ as $\rand_i,\sigma_i$. Set $m_i\equiv (\account,\rand_i)$
        \item If $\ds.\verify_{\vk}(m_i,\sigma_i)=0$, output $\bot$. \label{it:top-up-contains-ds}
        \item If $m_i=m_j$ for some $j<i$, output $\bot$.
    \end{compactenum}
    \item Else:
    \begin{compactenum}
        \item Interpret $\witness_i$ as $\pk_i,m_i\equiv (\account_{i,1},\ldots,\account_{i,r_i}), \sigma_i,t_i,\witness'_i$. 
        \item \label{it:balance_check_oss} If $\qs.\verify(\mpk,\pk_i,m_i,\sigma_i)=0$, output $0$. \hfill // Invalid quantum-signature.  
        \item If $\account_{i,t_i}\neq \account$, output $0$. \hfill // The target should be the current account. 
        \item \label{it:warm-up_balance_check_no_double_spend} 
        If $\account_i=\account_j$ and $t_i=t_j$ for some $j<i$, output $0$.  \hfill // Same witness claimed twice.  
        \item \label{it:balance_check_balance}If $\balance(\bpk,\account_i,\witness'_i) < t_i$, output $0$. \hfill // $\account_i$ has insufficient funds.         
    \end{compactenum}
\end{compactenum}
\item Output $v$.
\end{compactenum}

\caption{Algorithms for the warm-up quantum payment scheme.}
\label{fig:warm-up_scheme}
\end{figure}

An informal description of the warm-up scheme was presented on \cpageref{par:introduction_warm-up_scheme}, and the formal scheme is given in \cref{fig:warm-up_scheme}. A close inspection reveals that the completeness \cref{it:balance_respects_top-up,it:balance_respects_payments} hold. The additive polynomial blow-up property, \cref{it:additive_polynomial_blow_up} holds sinc{}e all the other pieces in $\pwit_i$ other than $\witness'$ (namely, $\account,m,i,\sigma$ --- see \cref{it:simple_pay_pwit} in $\simplepay$) are computed in time $\poly[\lambda+j_0]$, and therefore could have at most polynomial length. We emphasize that $\witness$ can be arbitrarily long. 
\begin{theorem}
The warm-up quantum payment scheme provided in \cref{fig:warm-up_scheme} is unforgeable, assuming $\ds$ is PQ-EU-CMA secure (see \cref{def:unforgeability_of_digital_signature}), and $\qs$ is unforgeable (see \cref{eq:ts_unforgeability,eq:sqts_unforgeability,eq:oss_unforgeability}).
\end{theorem}
\begin{proof}

We start with a simple combinatorial lemma. For a digraph $G=(V,E)$ we denote the in-degree of a vertex $v$ by $\deg^-(v)$ and the out-degree by $\deg^+(v)$.
\begin{lemma}
    Let $G=(V,E)$ be a finite digraph with a special node (which we call a source node) $s \in V$ and a target set $T \subseteq V \setminus \{s\}$. If $G$ satisfies: 
    \begin{compactenum}
         \item \label{it:positive_balance} For every $v \in V\setminus \{s\}$, $\deg^-(v)-\deg^+(v)\geq 0$,
         \item \label{it:zero_out_degree} For every $v \in T$, $\deg^+(v)=0$,
     \end{compactenum} 
     then,
    \begin{equation}
        \deg^+(s)\geq \sum_{v\in T} \deg^-(v).
    \end{equation}
\label{le:outdeg_s_geq_in_degrees_of_T}
\end{lemma}

 The proof outline is as follows. We later construct a digraph $G$ with a source node $s$ representing the challenger's top-up operations and a target set of nodes $T$ of size $m$ representing the challenger's accounts in the forgeability game (see \cref{line:target_accounts}). We show that this digraph satisfies the assumptions needed for \cref{le:outdeg_s_geq_in_degrees_of_T}, and some additional properties:

\begin{lemma}
    The following properties hold in $G$ with overwhelming probability:
    \begin{enumerate}
         \item \label{it:property_1_holds_in_G} For every $v \in V\setminus \{s\}$, $\deg^-(v)-\deg^+(v)\geq 0$.
         \item \label{it:property_2_holds_in_G} For every $v \in T$, $\deg^+(v)=0$.
        \item \label{it:in_degree_geq_b_i} For every target account $\{\account_i\}_{i \in \{1,\ldots,m\}}$ in $T$, $\deg^-(\account_i) \geq b_i$ (where $b_i$ is defined in \cref{line:calculating_balances} in Game~\ref{game:unforgeability}).
        \item \label{it:out_degree_s} $\deg^+(s)\leq n$ (where $n$ is defined in \cref{line:adversary_prepares_witnesses} in Game~\ref{game:unforgeability}).
    \end{enumerate}
    \label{le:properties_of_G_warm-up}
\end{lemma}

The theorem follows by combining the lemmas above, namely, that the following holds with overwhelming probability:
\begin{equation}
    n \stackrel{\cref{le:properties_of_G_warm-up}.\ref{it:out_degree_s}}{\geq}\deg^+(s) \stackrel{\cref{le:outdeg_s_geq_in_degrees_of_T},\cref{le:properties_of_G_warm-up}.\ref{it:property_1_holds_in_G},  \cref{le:properties_of_G_warm-up}.\ref{it:property_2_holds_in_G}}{\geq} \sum_{v\in T} \deg^-(v) = \sum_{i=1}^m \deg^-(\account_i) \stackrel{\cref{le:properties_of_G_warm-up}.\ref{it:in_degree_geq_b_i}}{\geq} \sum_{i=1}^m b_i.
\end{equation}

All that is left is proving these two lemmas.

\begin{proof}[Proof of \cref{le:outdeg_s_geq_in_degrees_of_T}]
Recall that $\sum_{v \in V} (\deg^+(v)-\deg^-(v)) = 0$. By rearranging, 
    \begin{align}
        \deg^+(s) &= \deg^-(s) + \sum_{v \in V \setminus \{s\}} (\deg^-(v) - \deg^+(v)) & \\ 
        &=\deg^-(s) + \sum_{v \in T} (\deg^-(v) - \deg^+(v)) +  \sum_{v \in V \setminus (T \cup \{s\}  )} (\deg^-(v) - \deg^+(v))& \\
        & \geq \sum_{v \in T} (\deg^-(v) - \deg^+(v)) & \text{(by \cref{it:positive_balance})} \\
        &=\sum_{v \in T} \deg^-(v), &\text{(by \cref{it:zero_out_degree})}
    \end{align}
which completes the proof of the lemma.
\end{proof}

\paragraph{The construction of $G$.} The digraph $G=(V,E)$ is constructed based on the witnesses supplied by the adversary in \cref{line:adversary_prepares_witnesses}. For the sake of the construction, we assume that the challenger maintains a data structure for creating and updating $G$. The operation $AddNode(v)$ adds $v$ to $V$. This operation would have no effect if $v$ was already added to the graph in an earlier step. Our digraph could have multiple directed edges between two nodes, and we use edge-labels to identify them. The operation $AddEdge(u,v,\ell)$ adds the edge $u,v$ with a label $\ell$ to $E$. We emphasize that this operation has no effect if $(u,v,\ell)$ was already added to the graph in an earlier step.  For analysis purposes, we assume that the challenger initializes an empty graph, adds a special source vertex $s$, and appends the following steps to the $\balance$ procedure (see \cref{fig:warm-up_scheme}) she uses:
\begin{compactenum}
  \setcounter{enumi}{4}
    \item  $AddNode(\account)$
    \item For every $i \in \{1,\ldots v\}$
    \begin{compactenum}
        \item If $\witness_i$ is a top-up witness, then $AddEdge(s,\account,r_i)$.
        \item Else: $AddEdge(\account_i,\account,(m_i,t_i))$.
    \end{compactenum}
\end{compactenum}
Clearly, this does not change the adversary's view, nor the winning probability in the game. 
We define the set $T$ to be the target accounts $(\account_i)_{i \in \{1,\ldots,m\}}$ generated by the challenger (see  Game \ref{game:unforgeability}, \cref{line:target_accounts}).

\begin{proof}[Proof of \cref{le:properties_of_G_warm-up}] \cref{it:property_1_holds_in_G}:
Let $\account'$ be a vertex in $V \setminus \{s\}$. We make two observations: i) Every two outgoing edges from $\account'$ labeled as $(m,t)$ and $(m',t')$ satisfy that $m=m'$ with overwhelming probability. Note that before adding such an edge, we check that $m$ was signed using $\qs$ (an uncloneable signature scheme) associated with $\account'$ (see \cref{it:balance_check_oss}). Therefore, having $m\neq m'$ as the outgoing edge-labels means that the adversary produced two  signed messages associated with $\account'$ that pass verification, violating the uncloneable signature unforgeability (see \cref{eq:ts_unforgeability,eq:sqts_unforgeability,eq:oss_unforgeability}), and therefore can only occur with negligible probability.
ii) An outgoing edge from $\account'$ labeled as $(m,t)$ implies that $\deg^-(\account')\geq t$. The implication follows from the recursive call to $\balance$ in \cref{it:balance_check_balance}, which adds at least $t$ in-coming edges to $\account'$. Indeed, we would not add edges with same label twice in that recursive call because of the test made in \cref{it:warm-up_balance_check_no_double_spend} in \cref{fig:warm-up_scheme}. 

We conclude that $\deg^-(\account') \geq \deg^+(\account')$: 
Let $t_{max}$ be the maximal $t$ for which there is an outgoing edge from $\account'$ labeled $(m,t)$. By the first observation, and the fact that all the labels are positive integers, we have $t_{max} \geq \deg^+(\account')$. By the second observation, we have $\deg^-(\account') \geq t_{max}$. Combining these inequalities completes the proof of \cref{it:property_1_holds_in_G}, namely, $\deg^-(\account')-\deg^+(\account')\geq 0$.

\cref{it:property_2_holds_in_G}: An outgoing edge is added to an $\account$ in $T$ only if the adversary can provide a signed message associated with it. We argue that such an edge would violate the security of the uncloneable signature scheme $\qs$, and therefore, $\deg^+(\account)=0$ for all the accounts in $T$, except with negligible probability: Recall that the set $T$ is the set of accounts created by the \emph{challenger} (see Game~\ref{game:unforgeability}), using the $\qs.\qgen$  algorithm (see $\createsimpleaccount$ in \cref{fig:warm-up_scheme}). By the unforgeability of the relevant $\qs$ scheme (see \cref{eq:ts_unforgeability,eq:sqts_unforgeability,eq:oss_unforgeability}), an outgoing edge for a vertex in $T$ can only occur with negligible probability; otherwise, we can construct a $\qs$ adversary, who simulates both the challenger and payment scheme adversary,  and provides one $\qs$ signature from the payment scheme adversary, and another from the challenger (who has $\qsk$, and therefore can be used to sign an arbitrary message).

\cref{it:in_degree_geq_b_i}. Recall that $b_i \gets \balance(\bpk,\account_i,\witness_i)$ and that we add $b_i$ incoming edges, during that call to $\balance$ (see the construction of $G$). Note that we might add even more incoming edges in the recurrence calls. Therefore, this property holds with certainty.

\cref{it:out_degree_s}. We add an outgoing edge from $s$ only for valid top-up witnesses. A top-up witness is valid only if it contains a valid digital signature  (see \cref{it:top-up-contains-ds} in \cref{fig:warm-up_scheme}), and the edge label contains the signed message. During the security game, exactly $n$ digital signatures are handed to the adversary, once for each $\topup$. Therefore, the out-degree is at most $n$ with overwhelming probability, or otherwise we can construct a \qpt adversary which would break the PQ-EU-CMA digital signature scheme, which contradicts our assumption.
\end{proof}
\end{proof}

\subsection{Divisibility and Mergeability}
\label{sec:divisibility_and_mergeability}
Any quantum payment scheme, such as the warm-up construction discussed previously in this section, provides divisibility and mergeability. To divide an account with balance $b$ associated with the state $\ket{\$}$ and $\witness$ to two accounts with balances $b_1$ and $b_2$ (where $b_1+b_2=b$), a user would create two accounts and move the appropriate amounts to them using $\simplepay$. 
After running this procedure, by Completeness \cref{it:balance_respects_payments}, we end up with having $\balance(\bpk,\ket{\$_i},\witness_i)=b_i$ for $i\in \{1,2\} $, as expected. We emphasize that the number of qubits in $\ket{\$} $ depends only on $\secpar$, and independent of all the other parameters, which is relevant in cases where quantum storage is scarce. Of course, the $\witness$ sizes, which could be viewed as the classical part of the money, increase in size after a division. Merging achieves the opposite goal: it takes as input two quantum money states with balances $b_1$ and $b_2$, and merge them into a single quantum money state with a balance of $b_1+b_2$, where the size of the quantum state depends only on $\secpar$.  These algorithms are shown if more detail in \cref{fig:divide_and_merge}. 

\begin{figure}[thbp]
\noindent$\dvd(\bpk,b_1,b_2,\ket{\$},\witness,\ket{\$_1},\account_1,\ket{\$_2},\account_2)$

\noindent  \textbf{Assumes:} $\ket{\$},\ket{\$_1},\ket{\$_2}$ are associated with $\account,\account_1,\account_2$, respectively.  \\
\begin{compactenum}
   \item If $b_1,b_2\notin \NN^+$ or $b_1+b_2\neq \balance(\bpk,\account,\witness)$, then Abort.
    \item Set $\vec \account=(\overbrace{\account_1,\account_1,\ldots,\account_1}^{b_1 \text{ times}},\overbrace{\account_2,\account_2,\ldots,\account_2}^{b_2 \text{ times}})$.
    \item $(\pwit_1,\ldots ,\pwit_b)\gets \simplepay(\bpk,\ket{\$},\account, \witness, \vec \account ) $.
    \item Set $\witness_1\equiv (\pwit_1,\ldots,\pwit_{b_1})$ and $\witness_2=(\pwit_{b_1+1},\ldots,\pwit_{b})$.
    \item Output $(\account_1,\ket{\$_1},\witness_1,\account_2,\ket{\$_2},\witness_2)$.
\end{compactenum}

\noindent $\merge(\bpk,\account_1,\ket{\$_1},\witness_1,\account_2,\ket{\$_2},\witness_2,\ket{\$},\account)$\\
\noindent  \textbf{Assumes:} $\ket{\$_1},\ket{\$_2},\ket{\$}$ are associated with $\account_1,\account_2,\account$, respectively.
 
\begin{compactenum}
\item $b_1 \gets \balance(\bpk,\account_1,\witness_1)$.
\item $b_2 \gets \balance(\bpk,\account_2,\witness_2)$.
\item If $b_1=0$ or $b_2=0$ then Abort.
\item Set $b:=b_1+b_2$.
\item $(\pwit_1,\ldots ,\pwit_{b_1})\gets \simplepay(\bpk,\ket{\$_1},\account_1, \witness_1, (\overbrace{\account,\account,\ldots,\account}^{b_1 \text{ times}}) ) $.
\item $(\pwit_{b_{1}+1},\ldots ,\pwit_{b})\gets \simplepay(\bpk,\ket{\$_1},\account_1, \witness_1, (\overbrace{\account,\account,\ldots,\account}^{b_2 \text{ times}}) ) $.
\item Set $\witness :=(\pwit_1,\ldots,\pwit_b) $.
\item Output $(\account,\ket{\$},\witness)$.
\end{compactenum}

\caption{Divide and Merge algorithms.}
\label{fig:divide_and_merge}
\end{figure}
 

\section{The main quantum payment scheme}
\label{sec:construction_payment_scheme}

The goal of this section is to extend the warm-up scheme. We will see in \cref{sec:applications} how it could be used to implement various prudent contracts. Similarly to the warm-up construction, our main construction is primarily based on uncloneable signatures and digital signatures. 

We will now discuss the differences between the warm-up scheme and the main scheme, see \cref{fig:prudent_contract_scheme,fig:prudent_contract_scheme_balance}. An account is associated with a single quantum signing token in the warm-up scheme. The main scheme supports multiple quantum signing tokens. Since there are many quantum signing tokens, we need to define a mechanism that determines their relative power: this is achieved by a program called $\siginterpreter$, specified by the user. This program receives as an input many messages, and outputs a single message, which represents the final decision. Each message is supposed to be signed by the associated quantum signing key. 

This interpreter can even decide if not all the messages are available. The role of $\siginterpreter$ is to interpret, even if some of the signatures are missing, the intended goal of the transaction. In the context of 3-out-of-5 multi-sig, the $\siginterpreter$ could output an interpreted message $m$ if at least 3-out-of-5 of the input messages are $m$. 

Since the $\siginterpreter$ influences where the funds are forwarded, there could be an attempt for double-spending. For example, consider the problematic setting in which there are two quantum signing keys, and in case one message is missing, the decision is based solely on the other message, which is present. This could pose a problem: Mallory could use her first quantum signing key to move her funds to Alice, and give that witness to her, and use the second quantum signing key to move the funds to Bob. Therefore, a $\siginterpreter$ has to be monotone, as follows:
\begin{definition}[Monotone sig-interpreter]
\label{def:siginterpreter}
\siginterpreter receives as an input an $n$-tuple of strings $(m_1,\ldots,m_n)$ where $m_i \in \{0,1\}^* \cup \visiblespace$, and deterministically outputs $\interpretedmessage \in\{0,1\}^* \cup \bot$.
The program \siginterpreter must satisfy the following monotonicity property:
Replacing the empty string $\visiblespace$ with a non-empty string does not change the output value from one valid value (i.e., any $m\neq \bot$) to another. Formally, for every $m_1,\ldots,m_n \in \{0,1\}^* \cup \visiblespace$ and every $i \in [n]$, if 
\[\siginterpreter(m_1,\ldots,m_{i-1},\visiblespace,m_{i+1},\ldots,m_n)\neq \bot,\]
then  
 \[\siginterpreter(m_1,\ldots,m_{i-1},\visiblespace,m_{i+1},\ldots,m_n) = \siginterpreter(m_1,\ldots,m_{i-1},m_i,m_{i+1},\ldots,m_n).\]
\end{definition}

\begin{remark}
Validating that $\siginterpreter$ is monotone could be computationally costly. In the schemes which we provide, we prove it directly. Two of the approaches that could be used are: 
\begin{enumerate}
    \item The scheme would support a fixed number of $\siginterpreter$ programs, which are guaranteed to be monotone. This somewhat reduces the flexibility.
    \item Whenever a $\siginterpreter$ program is specified, attached to it must be a formal proof that this program is monotone. It is the user's role to verify that correctness of that proof while running $\balance$. In this approach, some upper bounds on the running time of this verification should be made, so that verification time would not blow up. 
\end{enumerate}
\label{rem:monotonicity_could_be_hard_to_compute}
\end{remark}

The motivation for the definition above is the following.
\begin{lemma}
Let $\siginterpreter$ be a monotone sig-interpreter program, as in \cref{def:siginterpreter}. 
Fix $(m_1,\ldots,m_n),(m_1',\ldots,m_n')\in  \{0,1\}^* \cup \visiblespace$.
If  $\siginterpreter(m_1,\ldots,m_n),\siginterpreter(m_1',\ldots,m_n')$ $\neq \bot$ and
$\siginterpreter(m_1,\ldots,m_n)\neq \siginterpreter(m_1',\ldots,m_n')$ then there exists $i \in [n]$ such that $m_i\neq m_i'$ and $m_i,m_i'\neq \visiblespace$.    
\label{le:monotone_implies_existence_of_different_message}
\end{lemma}

This lemma's usefulness stems from the following reason. In our application, every $m_i \neq \visiblespace$ represents a message signed with the uncloneable signature scheme $\qs$ with respect to some $\pk_i$, whereas $m_i = \visiblespace$ represents a missing signed message. In this context, the lemma implies that an adversary cannot find two input vectors that cause $\siginterpreter$ to output two distinct valid outputs\footnote{By a valid output we mean any value which is \emph{not} $\bot$.},
since such an adversary could produce two signed messages, $m_i$ and $m_i'$ associated with the same public key $pk_i$, breaking the uncloneable signature security of $\qs$.
\begin{proof}
    We define 
    \begin{equation}
        A = \{i|m_i=\visiblespace \text{ and } m_i' \neq \visiblespace \},\quad A'=\{i|m_i\neq \visiblespace \text{ and } m_i' = \visiblespace \},\quad  B=\{i|m_i \neq m_i \text { and } m_i,m_i'\neq \visiblespace \}.
    \end{equation}
    Our goal is to prove that $B$ is non-empty. Suppose, by contradiction, that $B = \emptyset$. 
    Since, by assumption, $\siginterpreter(m_1,\ldots,m_n)\neq \siginterpreter(m_1',\ldots,m_n')$, then either $A$ or $A'$ must be non-empty. Let 
    \begin{equation}
         x_i=
    \begin{cases}
        m'_i &\text{if } i\in A\\ 
        m_i & \text{otherwise.} \\
          \end{cases}
     \end{equation} 
     \begin{equation}
         x_i'=
    \begin{cases}
        m_i &\text{if } i\in A'\\ 
        m_i' & \text{otherwise.} \\
          \end{cases}
     \end{equation} 
     Since by assumption $\siginterpreter(m_1,\ldots,m_n)\neq \bot$, by the monotonicity property (see \cref{def:siginterpreter}), 
     
     $\siginterpreter(x_1,\ldots,x_n)=\siginterpreter(m_1,\ldots,m_n)$,
     and similarly, 
     
     $\siginterpreter(x_1',\ldots,x_n')=\siginterpreter(m_1',\ldots,m_n')$. 
     But note that whenever $B=\emptyset$, by construction, $(x_1,\ldots,x_n)=(x_1',\ldots,x_n')$. 
     This gives us a contradiction, since on the one hand, by combining these facts, we have
     \begin{align}
     \siginterpreter(m_1,\ldots,m_n)&=\siginterpreter(x_1,\ldots,x_n)\\
     &=\siginterpreter(x_1',\ldots,x_n')\\
     &=\siginterpreter(m_1',\ldots,m_n'),
     \end{align}
     while by assumption $\siginterpreter(m_1,\ldots,m_n)\neq \siginterpreter(m_1',\ldots,m_n') $.
\end{proof}

The $\siginterpreter$ can also act as a financial control mechanism. For financial control purposes, the interpreter could define a simple mapping, where upon input message $i$, the money would be transferred to $\account_i$, for some predefined set of accounts. This would mean that money can only be sent to one of these predefined accounts. One motivation is to reduce embezzlement risk; another is to reduce the risk of theft.

Recall that in the warm-up scheme, the sender signs the accounts to which she wants to send the money. In the main scheme, this is done differently. The interpreter outputs a program, denoted $\outputscript$, which determines where the funds will be transferred: upon input $k$, the program outputs the $k$th destination account. But one might wonder, are we just saving space? When the sender has a balance of $b$, why can't the signer specify $b$ destinations? In some cases, the signer may want to decide now, regarding the destination that are not currently available. 
Consider, for example, the permanent accounts example that was discussed.
Here, the account might receive funds even \emph{after} the quantum signing key is used. There is no upper-bound on the number of payments which will be made to that address in the future. In this case, the spender would first create a new address, send the existing $b$ funds to their destinations, in a similiar fashion as was done in the warm-up scheme, and for any $k>b$, the $k$'th unit will be transferred to the new account. This is somewhat analogous to call-forwarding: all remaining balance would be automatically moved to the new account. The main advantage  of the program over a list here is that a (finite) program can implicitly define an infinite list.

The next new ingredient is the $\outputverifyscript$. This is another program that the $\siginterpreter$ outputs. This program can block some of the outputs determined by $\outputscript$ from taking effect. Applications which use this functionality are \emph{restricted accounts} in \cref{sec:restricted_accounts} and \emph{secure exchange} in \cref{sec:secure_exchange}.

The last new ingredient is a parameter for auxiliary information denoted by $\aux$---see \cref{it:interpretedmessage_aux} in \cref{fig:prudent_contract_scheme_balance}. A user may choose to sign $\aux$ with the uncloneable signature $\qs$, which is completely ignored by all the algorithms of the payment scheme. This can be useful for extensions of the scheme, such as the applications discussed in \cref{sec:colored_coins_smart_property_and_tokenized_securities}: there we show an application which piggybacks the payment scheme and use the auxiliary information as a way to authenticate messages (specifically, to allow shareholders to exercise their voting rights). This ingredient is also used in \cref{sec:proof_of_reserves}.

In a simple account, the account is a single \qs public key and a $\siginterpreter$ program, which in this case is the identity program (i.e., it outputs the input that it received). The  $\account$ would consist of multiple \qs public keys in more advanced use-cases. 

Our main payment scheme is presented in full detail in \cref{fig:prudent_contract_scheme,fig:prudent_contract_scheme_balance}.

\begin{figure}[th!]
\noindent\textbf{Assumes:}  Digital Signature $\ds$, and $\qs$ is an uncloneable signature scheme.
\\
\noindent$\setup(\secparam)$ \hfill // The same as in the warm-up scheme. 
\begin{compactenum}
\item $(\sk,\vk) \gets \ds.\keygen(\secparam)$.
\item $(\mpk,\msk) \gets \qs.\setup(\secparam)$. 
\item $\bpk\equiv (\mpk,\vk)$, $\bsk \equiv (\msk,\sk)$.
\item Output $(\bpk,\bsk)$.
\end{compactenum}

\noindent $\simplesiginterpreter(m_1)$
\begin{compactenum}
\item Output $m \equiv m_1$.
\end{compactenum}

\noindent $\createsimpleaccount(\bsk)$ \hfill // If one-shot signatures are used, $\bpk$ is used instead of $\bsk$. 
\begin{compactenum}
\item Interpret $\bsk$ as $(\msk,\vk)$. 
\item $(\pk, \qsk) \gets \qs.\qgen(\msk)$.\hfill // If $\qs$ is a one-shot signature scheme, $\msk$ should be replaced with $\mpk$.
\item $\ket{\$}\equiv \qsk,\ \account \equiv (\pk,\simplesiginterpreter)$.
\item Output $ (\ket{\$},\account)$.
\end{compactenum}

\noindent $\topup_\bsk(\account)$ \hfill // The same as in the warm-up scheme. 
\begin{compactenum}
\item Interpret $\bsk$ as $(\msk,\sk)$.
\item Let $\rand\gets_R\{0,1\}^\lambda$.
\item Set $m\equiv (\account,\rand)$.
\item $\sigma \gets \ds.\sign_{\sk}(m)$.
\item Output $\pwit \equiv (\rand,\sigma) $.
\end{compactenum}

\noindent $\simpleoutputscript_{(\account_1,\ldots,\account_n)}(i)$
\begin{compactenum}
\item If $1 \leq i \leq n $, output $\account_i$, otherwise, output $\bot$. 
\end{compactenum}

\noindent $\simpleoutputverifyscript(i,\vswitness)$
\begin{compactenum}
\item Output True.
\end{compactenum}

\noindent $\simplepay(\bpk, \ket{\$},\account, \witness,(\account_1,\ldots,\account_r) )$ \\
\hfill Assumes that $\balance(\bpk,\account,\witness)=r$. 
\begin{compactenum}
\item \label{it:simplepay_m}Set $m \equiv (\simpleoutputscript_{\account_1,\ldots,\account_r},\simpleoutputverifyscript)$.
\item $\sigma \gets \qs.\qsign(\ket{\$} ,m)$.
\item For every $i\in [r]$:
\begin{compactenum}
    \item Set $\pwit_i \equiv (\account, m, \sigma,i,NULL, \witness)$. \hfill // We do not need a $\vswitness$ for simple accounts, hence the NULL value. 
\end{compactenum}
\item Output $(\pwit_1,\ldots,\pwit_r) $.
\end{compactenum}
\caption{Algorithms for the quantum payment scheme.}
\label{fig:prudent_contract_scheme}
\end{figure}

\begin{figure}[tb]
\noindent $\balance(\bpk, \account,\witness)$
\begin{compactenum}
    \item Interpret $\bpk$ as $(\mpk,\vk)$.
    \item Interpret $\witness$ as $(\witness_1,\ldots, \witness_v)$. \hfill // $v$ is the claimed balance. 
    \item For every $i\in\{1,\ldots,v\}$:
    \begin{compactenum}
        \item If $\witness_i.size=2$: \hfill // A top-up witness, no changes from the warm-up scheme.  \label{it:full-top-up-witness}
        \begin{compactenum}
            \item Interpret $\witness_i$ as $(\rand_i,\sigma_i)$. Set $m_i\equiv (\account,\rand_i)$.
            \item If $\ds.\verify_{\vk}(m_i,\sigma_i)=0$, output $\bot$. \label{it:full-top-up-contains-ds}
            \item If $m_i=m_j$ for some top-up witness $j<i$, output $0$.
        \end{compactenum}
        \item Else:
        \begin{compactenum}
            \item Interpret $\witness_i$ as $(\account_i, \vec m_i,t_i, \vec\sigma_i,\vswitness_i, \witness'_i)$.
            \item Interpret $\account_i$ as $(\pk_{i,1},\ldots,\pk_{i,n_i},\siginterpreter_i)$. 
            \item \label{it:definition_of_vec_m}Interpret $\vec m_i$ as $m_{i,1},\ldots,m_{i,n_i}$ where $m_{i,j}\in \{0,1\}^*\cup \visiblespace$ and $\vec \sigma_i$ as $\sigma_{i,1},\ldots,\sigma_{i,n_i}$ where $\sigma_{i,j}\in \{0,1\}^*\cup \visiblespace$.  
            \item Output $0$, unless the all the following tests succeed: 
            \begin{compactenum}
                \item \label{it:siginterpreter_is_monotone_test} $\siginterpreter$ is monotone (\cref{def:siginterpreter}). \hfill // See \cref{rem:monotonicity_could_be_hard_to_compute}. 
                \item \label{it:valid_message_is_signed} For every $j\in \{1,\ldots,n_i\}$ such that $m_{i,j} \neq \visiblespace$, $\qs.\verify(\mpk,\pk_{i,j},m_{i,j},\sigma_{i,j})=1$.
                \item \label{it:interpreated_message_valid} Set $\interpretedmessage_i\gets \siginterpreter(m_{i,1},\ldots,m_{i,n_i})$. Test that $\interpretedmessage_i \neq \bot$. 
                \item \label{it:interpretedmessage_aux} Interpret $\interpretedmessage_i$ as $(\outputscript_i,\outputverifyscript_i,\aux_i)$. 
                \item $\outputscript_i(t_i)= \account$. \hfill // The $t_i$th output is the current $\account$. 
                \item \label{it:main_checks_outputverifyscript}$\outputverifyscript_i(t_i,\vswitness_i)=1$. \hfill // There is a verify-script witness which allows it to be paid. 
                \item \label{it:main_balance_check_no_double_spend} $(\account_i,t_i)\neq (\account_j,t_j)$ for all $j < i$. \hfill // The same witness should not be claimed twice. 
                \item \label{it:balance_checks_balance_main} $\balance(\bpk,\account_i,\witness'_i)\geq t_i$. \hfill // $\account_i$ has sufficient funds. 

            \end{compactenum}
        \end{compactenum}
    \end{compactenum}
    \item Output $v$.
\end{compactenum}
\caption{The balance algorithm in our quantum payment scheme.}
\label{fig:prudent_contract_scheme_balance}
\end{figure}

\begin{theorem}
The quantum payment scheme in \cref{fig:prudent_contract_scheme,fig:prudent_contract_scheme_balance} is unforgeable, assuming $\ds$ is post-quantum EU-CMA secure (see \cref{def:digital_signature}), and $\qs$ is unforgeable (see \cref{eq:ts_unforgeability,eq:sqts_unforgeability,eq:oss_unforgeability}).
\end{theorem}

\begin{proof}
The outline of the proof is the same as the unforgeability proof in the warm-up-scheme: we construct a labeled graph $G$ in the same way as was done for the warm-up scheme, except we use the label $(\interpretedmessage_i,t_i)$ instead of $(m_i,t_i)$; we show that $G$ satisfies the same  properties (see \cref{le:properties_of_G_warm-up}); and use the combinatorial lemma (see \cref{le:outdeg_s_geq_in_degrees_of_T}) to complete the proof.

The only missing part to prove is the following: 
\begin{lemma}
    The following properties hold in $G$ with overwhelming probability:
    \begin{enumerate}
         \item \label{it:property_1_holds_in_G_main} For every $v \in V\setminus \{s\}$, $\deg^-(v)-\deg^+(v)\geq 0$.
         \item \label{it:property_2_holds_in_G_main} For every $v \in T$, $\deg^+(v)=0$.
        \item \label{it:in_degree_geq_b_i_main} For every target account $\{\account_i\}_{i \in \{1,\ldots,m\}}$ in $T$, $\deg^-(\account_i) \geq b_i$ (where $b_i$ is defined in \cref{line:calculating_balances} in Game~\ref{game:unforgeability}).
        \item \label{it:out_degree_s_main} $\deg^+(s)\leq n$ (where $n$ is defined in \cref{line:adversary_prepares_witnesses} in Game~\ref{game:unforgeability}).
    \end{enumerate}
    \label{le:properties_of_G_main} 
\end{lemma}
\begin{proof}
The proofs of \cref{it:property_2_holds_in_G_main,it:in_degree_geq_b_i_main,it:out_degree_s_main} are exactly the same as in \cref{le:properties_of_G_warm-up}, and hence omitted. We are left to prove \cref{it:property_1_holds_in_G_main}:
Let $\account'$ be a vertex in $V \setminus \{s\}$. We make the similar two observations that were made previously: 

i) Every two outgoing edges from $\account'$ labeled as $(\interpretedmessage,t)$ and

$(\interpretedmessage',t')$ satisfy that $\interpretedmessage=\interpretedmessage'$ with overwhelming probability.
Suppose toward the contradiction that $\interpretedmessage\neq \interpretedmessage'$.
Denote by $\vec m$ $(\vec m')$ and $\vec \sigma$ ($\vec \sigma'$) defined in \cref{it:definition_of_vec_m} that were used to generate $\interpretedmessage$ (respectively $\interpretedmessage'$).
Since $\siginterpreter$ is monotone by the test in \cref{it:siginterpreter_is_monotone_test}, and $\interpretedmessage, \interpretedmessage' \neq \bot$ by the test in \cref{it:interpreated_message_valid} (see \cref{fig:prudent_contract_scheme_balance}), the conditions for \cref{le:monotone_implies_existence_of_different_message} hold, and therefore, there exists an $i$ such that $m_i \neq m_i'$ and $m_i,m_i' \neq \visiblespace$.
This can only occur with negligible probability, since in \cref{it:valid_message_is_signed} we check that $(\sigma_i,m_i)$ and $(\sigma_i',m_i') $ are valid \qs signatures associated with the same $\pk_i$, hence violating the \qs unforgeability guarantees (see \cref{eq:ts_unforgeability,eq:sqts_unforgeability,eq:oss_unforgeability}).  

ii) An outgoing edge from $\account'$ labeled as $(\interpretedmessage,t)$ implies that

$\deg^-(\account')\geq t$. As in the warm-up scheme, this implication follows from the recursive call to $\balance$ in \cref{it:balance_checks_balance_main}, which adds at least $t$ in-coming edges to $\account'$. Indeed, we would not add the edges with same label twice in that recursive call because of the test made in \cref{it:main_balance_check_no_double_spend} in \cref{fig:prudent_contract_scheme_balance}. 

We conclude that $\deg^-(\account') \geq \deg^+(\account')$ by the exact same argument used in the warm-up scheme, which is repeated here for completeness. Let $t_{max}$ be the maximal $t$ for which there is an outgoing edge from $\account'$ labeled $(\interpretedmessage,t)$. By the first observation, and the fact that all the labels are positive integers, we have $t_{max} \geq \deg^+(\account')$. By the second observation, we have $\deg^-(\account') \geq t_{max}$. Combining these inequalities completes the proof of \cref{it:property_1_holds_in_G}, namely, $\deg^-(\account')-\deg^+(\account')\geq 0$. \end{proof}

This completes the security proof of the payment scheme.
\end{proof}

\section{Applications} 
\label{sec:applications}
\subsection{Permanent Accounts} 
\label{sec:permanent_accounts}

Here we discuss how to create and use a \emph{permanent} account. Recall that when executing the $\simplepay$ algorithm, the payer must specify all the outputs. A problem could occur if future payments were made to that address. If we were to use the $\simplepay$ algorithm, these future funds would be lost. 

This problem can be easily fixed. We can still use $\createsimpleaccount$ to generate a permanent account. We need to replace $\simplepay$. The idea is to provide a forwarding account: the spender first executes $\createsimpleaccount$ to generate the new $\forwardaccount$, to which all future payments would be forwarded.  Instead of using $\simpleoutputscript$ to spend the funds, we send all future outputs (i.e., outputs greater than $n$, where $n$ is the amount of funds that were received so far) to $\forwardaccount$.   
The algorithm for paying from a permanent account is given in \cref{fig:permanent_account}.

\begin{figure}[tb]
\noindent $\permanentoutputscript_{(\account_1,\ldots,\account_n)}(i)$
\begin{compactenum}
\item If $1 \leq i \leq  n-1$, output $\account_i$. 
\item Else, output $\account_n$.
\end{compactenum}

\noindent $\payfrompermanent(\bpk, \ket{\$},\account, \witness,(\account_1,\ldots,\account_r),\uwave{\forwardaccount} )$ \hfill // Should be compared with $\simplepay$ in \cref{fig:prudent_contract_scheme}.  \\
\hfill Assumes that $\balance(\bpk,\account,\witness)=r$. 
\begin{compactenum}
\item Set $m \equiv ($ \sout{$\simpleoutputscript_{\account_1,\ldots,\account_r}$} \uwave{$\permanentoutputscript_{\account_1,\ldots,\account_r,\forwardaccount}$} $,\simpleoutputverifyscript)$.
\item $\sigma \gets \qs.\qsign(\ket{\$} ,m)$.
\item For every $i\in [r]$:
\begin{compactenum}
    \item Set $\pwit_i \equiv (\account, m, \sigma,i,NULL, \witness)$.
\end{compactenum}
\item Output $(\pwit_1,\ldots,\pwit_r). $ \hfill // If future payments are made to $\account$, so that $\balance(\bpk,\account,\witness')>r$ they would be implicitly transferred to $\forwardaccount$, with $\pwit_i=(\account,m,\sigma,i,NULL,\witness')$ for $i > r$. 
\end{compactenum}  
\caption{Paying from a permanent account. The changes from $\simplepay$ (see \cref{fig:prudent_contract_scheme}) are marked.}
\label{fig:permanent_account}
\end{figure}


\subsection{Secure Exchange} 
\label{sec:secure_exchange}

Consider the setting where Alice wants to exchange her quantum euros for Bob's quantum pounds. Here, we assume that each one of the parties can cheat. If Alice sends Bob the money \emph{first}, then she is at risk that Bob will run away with the money. The goal here is to construct a protocol in which a dishonest party cannot run away with the honest party's money and their own. Note that in our protocol a dishonest party \emph{can} lock the honest party's funds, but \emph{cannot} run away with it. A \emph{cross-chain atomic swap} provides a stronger guarantee: an honest user (Alice or Bob in this case) would always finish with either the euro or the pound; in other words, the honest user funds cannot be locked. See also the discussion in \cref{sec:drawbacks}, \cref{it:atomic-swaps}.

We first describe the idea at a high level. The construction uses a post-quantum one-way function $h$. Alice picks a random $x$ and computes $y=h(x)$.
She then creates a transaction that moves her funds to Bob's account. She sends all the information about the transaction which would make it valid, \emph{except} for $x$: her transaction includes a $\outputverifyscript$ restriction which requires a $\vswitness$ of the form $x'$ such that $h(x')=y$. Clearly, $x$ satisfies this property.  This restriction is given in the following script (where $\exchangeoutputverifyscript$ stands for Exchange Output Verify Script):
\procedureblock{$\exchangeoutputverifyscript_y(i,x)$}{%
    \text {Accept iff } h(x)=y \\
} 

 So, at this point, Bob would like Alice to reveal $x$. After making sure that the only missing piece is $x$, Bob sends Alice his pounds using $\simplepay$. Alice verifies that the transaction is valid and, if so, sends $x$ to Bob. At this point, both parties hold a valid witness for the funds, and the protocol terminates successfully.

We also define the function $\balance'$ which is the same as $\balance$ except it ignores the $\outputverifyscript$ test in \cref{it:main_checks_outputverifyscript} \emph{only} in the non-recursive calls. $\balance'$ is used by Bob before he receives the preimage of the one-way function. 

The protocol is given in \cref{fig:secure_exchange}.  

\begin{figure}[thbp]
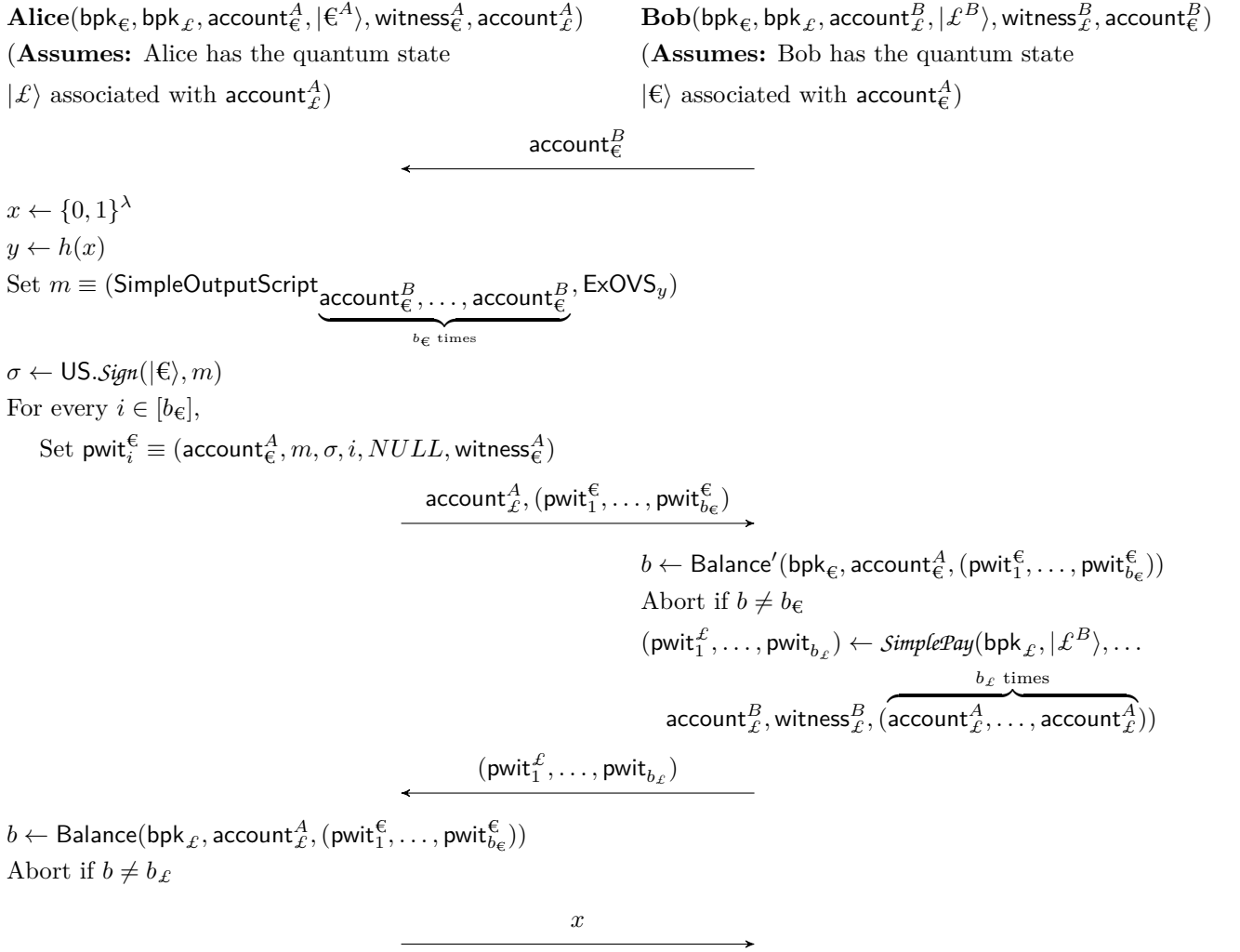

\procedureblock[colspace=-4.2cm]{}{%
\textbf{Alice}(\bpk_\eur,\bpk_\str,\account^A_\eur,\ket{\eur^A},\witness^A_\eur,\account^A_\str) \> \> \textbf{Bob}(\bpk_\eur,\bpk_\str,\account^B_\str,\ket{\str^B},\witness^B_\str,\account^B_\eur) \\
\text{(\textbf{Assumes:} Alice has the quantum state} \> \>\text{(\textbf{Assumes:} Bob has the quantum state}\\
 \text{$\ket{\str}$ associated with $\account^A_\str$)} \> \>\text{$\ket{\eur}$ associated with $\account^A_\eur$)}\\
\> \sendmessageleft*[5cm]{\account^B_\eur} \hspace{2cm} \> \\
x \gets \{0,1\}^\secpar \> \>  \\
y \gets h(x) \> \>  \\
\text{Set } m \equiv (\simpleoutputscript_{\underbrace{\account_\eur^B,\ldots,\account_\eur^B}_{b_\eur \text{ times}} },\exchangeoutputverifyscript_y) \> \>  \\
\sigma \gets \qs.\qsign(\ket{\eur},m) \> \>  \\
\text{For every } i\in [b_\eur],  \> \>  \\
\quad  \text{ Set } \pwit_i^\eur\equiv(\account^A_\eur,m,\sigma,i,NULL,\witness^A_\eur)  \> \>  \\
\> \sendmessageright*[5cm]{\account^A_\str,(\pwit_1^\eur,\ldots,\pwit_{b_\eur}^\eur)} \hspace{2cm}  \> \\ 
  \> \> b \gets \balance'(\bpk_\eur,\account^A_\eur,(\pwit_1^\eur,\ldots,\pwit_{b_\eur}^\eur) ) \\
  \> \> \text{Abort if } b\neq b_\eur   \\
  \> \> (\pwit_1^\str,\ldots,\pwit_{b_\str})\gets \simplepay(\bpk_\str,\ket{\str^B}, \ldots \\
  \> \> \quad \account^B_\str,\witness^B_\str,(\overbrace{\account_\str^A,\ldots,\account_\str^A}^{b_\str \text{ times}}  )) \\
\> \sendmessageleft*[5cm]{(\pwit_1^\str,\ldots,\pwit_{b_\str})} \hspace{2cm}  \> \\ 
 b \gets \balance(\bpk_\str,\account^A_\str,(\pwit_1^\eur,\ldots,\pwit_{b_\eur}^\eur) )\> \>  \\
 \text{Abort if } b\neq b_\str  \> \>  \\
 \> \sendmessageright*[5cm]{x} \hspace{2cm}  \> \\ 
 } 
\caption{Secure Exchange Protocol for exchanging $b_\eur$ quantum euros for $b_\str$ quantum sterlings.}
\label{fig:secure_exchange}
\end{figure}

The completeness of the protocol can be easily verified. The security can be analyzed as follows. After the first message from Alice, Alice's money is already locked in Bob's account (or otherwise, Bob would abort), and therefore she cannot run away with Bob's money. After the first message from Alice, Bob cannot runaway with Alice's money since he needs a pre-image of $y$, which is unfeasible due to the one-waynesss property of the one-way function $h$. Note that after Bob's second message, the protocol is aborted by Alice, unless Bob's funds are moved to her. Overall, Bob cannot runaway with Alice's money.  

Another classical approach to reduce the risk is to split the transaction into smaller parts. For example, suppose Alice and Bob agree to exchange her $\eur100$ for his $\str100$.
They could split the transaction to 100 exchanges of $\eur1$ for $\str1$, and, abort if at one point other party cheats. This guarantees that the gain from cheating is at most $\eur1$. 

We can combine the secure exchange protocol with the approach above. Note that here no party has an incentive to cheat, and the worst-case scenario is of $\$1$ getting locked (rather than stolen).

Interestingly, the secure exchange with splitting works with other fungible goods, such as the colored coins and stocks, and without the splitting for non-fungible goods such as the smart property example, which are all discussed in \cref{sec:colored_coins_smart_property_and_tokenized_securities}.
Furthermore, a very similar technique could be used to exchange Alice's quantum money with Bob's bitcoin (or any other similar cryptocurrency). The difference in the protocol are: 
\begin{enumerate}
     \item Instead of calling $\createsimpleaccount$, Alice would generate a Bitcoin secret key and send its public key (also known as a Bitcoin address) to Bob.
     \item Bob would send the payment to Alice using the standard Bitcoin payment protocol instead of  $\simplepay$.
     \item Alice would check that her Bitcoin balance is what she expects using her Bitcoin client, rather than the $\balance$ function. 
 \end{enumerate} 

 Interoperability between different blockchains has been studied extensively; see the recent survey in Ref.~\cite{BVGC22}.
 The example above can be viewed as the first example demonstrating interoperability between quantum money schemes (and quantum payment schemes in particular) and cryptocurrencies. 

\subsection{Backup and multiple signatures} 
\label{sec:backup_and_multiple_signatures}
One of the prominent techniques for internal control, as well as a backup in cryptocurrencies, is \emph{multisignatures}. For example, a Bitcoin address could be associated with $3$ signing keys so that any $2$ of these addresses could spend those Bitcoins. This is useful for backup purposes: the user could save the $3$ different signing keys in $3$ different locations so that losing one of them would not have disastrous consequences. This is useful for internal control purposes as well: in a firm, the accountant, CEO and founder may each create a signing key, so that every $2$ of them may be able to spend the money. This reduces the risk of embezzlement since this way it would require some collusion. For more details, see ~\cite[Section 4.3]{NBF+16}. The goal of this section is to provide such a multi-signature functionality. 

\begin{remark}
A simpler approach to achieving related goals is to use a secret sharing scheme to share a single secret key. This approach is useful for backup purposes but not as useful for internal control, since the distribution and reconstruction phases constitute a single point of failure. Therefore, we do not discuss this approach any further. For more details, see ~\cite[Section 4.3]{NBF+16}.
\end{remark}

In cryptocurrencies, we can construct a scheme in which any $1$-out-of-$2$ signing keys is enough for spending. This cannot work in a quantum payment scheme since it would allow Mallory to spend her money twice: she could send her funds to Alice using her first \qs token and Bob using the second \qs token. In this context, the public ledger in cryptocurrencies provides a significant advantage, which we cannot use.

Nevertheless, we can construct multi-sig accounts for intersecting families, defined as follows:
\begin{definition}[Intersecting family (see, e.g., \cite{Juk11})] A family $\mathcal F$ is \emph{intersecting} if for every $X,Y\in \mathcal F$, $X \cap Y \neq \emptyset$.  We say that the family is \emph{over} the set $S$ if for every $X\in \mathcal F$, $X \subset S$.
\label{def:intersecting_family}
\end{definition}

\begin{example}
    For every $n\in \NN$ and $k> n/2$, the family $\mathcal F=\{X\subset [n]: |X|\geq k \}$ is an intersecting family. For every $n \in \NN$ and $k \leq n/2$, the family $\mathcal F=\{X\subset [n]: |X|\geq k \}$ is not an intersecting family. 
\end{example}

\begin{figure}[thbp]
\noindent$\multisiginterpreter_\intersectingfamily(m_1,\ldots,m_n)$\\
\noindent \textbf{Assumes:} $\intersectingfamily$ is an intersecting family over $[n]$.
\begin{compactenum}
   \item For every $S \in \intersectingfamily$
   \begin{compactenum}
       \item If $m_i=m_j\neq \bot $ for every $i,j\in S$, then Output $m_i$.
   \end{compactenum}
   \item Output $\bot$.
\end{compactenum}

\noindent$\createmultisigaccount (\bsk,\intersectingfamily)$ \hfill // If one-shot signatures are used, $\bpk$ is used instead of $\bsk$. \\
\noindent \textbf{Assumes:} $\intersectingfamily$ is an intersecting family over $[n]$.
\begin{compactenum}
   \item Interpret $\bsk$ as $(\msk,\sk)$.
   \item For $i=1,\ldots, n$:
   \begin{compactenum}
       \item $(\pk_i,\qsk_i) \gets \qs.\qgen(\msk)$. \hfill // If $\qs$ is a one-shot signature scheme, $\msk$ should be replaced with $\mpk$. 
   \end{compactenum}
   \item Set $\ket{\$}\equiv \qsk_1 \tensor \ldots \tensor \qsk_n$ and $\account \equiv (\pk_1,\ldots,\pk_n,\multisiginterpreter_\intersectingfamily) $.
   \item Output $(\ket{\$},\account)$.
\end{compactenum}

\noindent $\multisigpay_\intersectingfamily(\bpk,\account,S,\{\qsk_i\}_{i \in S}, \witness, (\account_1,\ldots,\account_r))$ \\
\noindent  \textbf{Assumes:} $\intersectingfamily$ is an intersecting family over $[n]$, and that for every $i \in S$, $\qsk_i$ is associated with $\account_i$, and that $\balance(\bpk,\account,\witness)=r$. 
\begin{compactenum}
    \item If $S \notin \intersectingfamily$, Output $\bot$.
    \item For every $i\in S$:
    \begin{compactenum}
        \item Set $m_i \equiv (\simpleoutputscript_{\account_1,\ldots,\account_r},\simpleoutputverifyscript,\visiblespace)$.   
        \item $\sigma_i \gets \qs.\qsign(\qsk_i,m_i) $.
    \end{compactenum}
    \item For every $i\in [n]\setminus S$, set $m_i \equiv \visiblespace$ and $\sigma_i\equiv \visiblespace$.
    \item For every $k \in [r]$, set $\pwit_k\equiv (\account,\vec{m}, \vec{\sigma},k,\visiblespace, \witness)$.
\end{compactenum}
\caption{Multi-sig account algorithms.}
\label{fig:multisig_algorithms}
\end{figure}

\cref{fig:multisig_algorithms} contains our construction for a multi-sig account and the accompanying algorithms. Note that the construction is valid for any intersecting family $\mathcal F$. Therefore, we can have a 2-of-3 multi-sig account but cannot have a 1-of-3 multi-sig account. Completeness can be verified easily. The only property that remains to be proved is that $\multisiginterpreter$ is monotone.
\begin{lemma}
    For any intersecting family $\mathcal F$ over $[n]$, the program $\multisiginterpreter_\mathcal F$ is a monotone sig-interpreter (see \cref{def:siginterpreter}).  
\end{lemma}
\begin{proof}
Suppose toward the contradiction, that there exists $m_1,\ldots,m_n \in \{0,1\}^* \cup \visiblespace$ and $i$ such that 
\[\multisiginterpreter(m_1,\ldots,m_{i-1},\visiblespace,m_{i+1},\ldots,m_n) = \interpretedmessage \neq \bot.\]
and 
 \[\multisiginterpreter(m_1,\ldots,m_{i-1},m_i,m_{i+1},\ldots,m_n) = \interpretedmessage' \neq \interpretedmessage.\]

 Let $S$ ($S'$) be the set in which the for loop in $\multisiginterpreter(m_1,\ldots,m_{i-1},\visiblespace,m_{i+1},\ldots,m_n)$ (respectively, $\multisiginterpreter(m_1,\ldots,m_{i-1},m_i,m_{i+1},\ldots,m_n)$) terminated with a non $\bot$ value. Note that $S$ is well defined, since we assume that
 
 $\multisiginterpreter(m_1,\ldots,m_{i-1},\visiblespace,m_{i+1},\ldots,m_n) = \interpretedmessage \neq \bot$. That also implies that $i \notin S$. 
 Consequently, $\multisiginterpreter(m_1, \ldots,\allowbreak m_{i-1}, m_i, m_{i+1}, \ldots, m_n)$ cannot terminate with a $\bot$ value, since $S$ would already cause $\multisiginterpreter$ to terminate (since $i \notin S$), and therefore $S'$ is also well defined.

 Since $\mathcal F$ is an intersecting family, and $S,S' \subseteq \mathcal F$, we know $S\cap S' \neq \emptyset$. Let $j\in [n]$ such that $j \in S\cap S'$. 
 Since $j\in S$, by construction, 
 \begin{equation}
  \interpretedmessage =\multisiginterpreter(m_1,\ldots,m_{i-1},\visiblespace,m_{i+1},\ldots,m_n)=m_j.
 \end{equation} 
 Similarly, since $j\in S'$, 
 \begin{equation}
    \interpretedmessage = \multisiginterpreter(m_1,\ldots,m_{i-1},m_i,m_{i+1},\ldots,m_n)=m_j. 
 \end{equation}
We reach the desired contradiction: $\interpretedmessage=\interpretedmessage'$.
\end{proof} 

\begin{remark} Some non-custodial Bitcoin wallets\footnote{For example, Electrum, a Bitcoin wallet, supports 2FA via a service offered by TrustedCoin.com, see \url{https://electrum.readthedocs.io/en/latest/2fa.html} and \url{https://trustedcoin.com/\#/faq}.} provide a mechanism for second-factor authentication (2FA), to provide better key-management. 2FA is achieved by using a 2-of-3 multi-sig Bitcoin address, in which the server holds one secret key, the user stores another secret key on his device, and a third secret key is stored by the user, say, on a piece of paper, and is kept for emergencies. One of the challenges with key-management is that the user's device, which stores one of the secret keys, might be hacked at some point. When using a 2FA service, the server signs a message remotely, only upon some predefined sets of criterion. The server often sends the transaction details through another channel (e.g. a text-message to the user's mobile phone), so the user could see the transaction details, and the server asks the user to confirm the transaction. In this example, not only the device containing the key should be hacked, but also the user's mobile phone. Other criterion, such as spending limits (for e.g., no more than \$1000 a day) could be enforced by the server. Most importantly, such a service is non-custodial: the user is not harmed even if the server is off-line, or even malicious, as long as the user has his recovery key. We can use a similar design using the multi-sig account presented above: we would use a 2-out-of-3 multi-sig account, one quantum signing key would be created by the server, and the two others kept by the user. 
\end{remark} 


\subsection{Restricted accounts} 
\label{sec:restricted_accounts}

A restricted account allow the account creator to restrict the destination accounts to which payments could be made. Here is one motivating example. A parent could create an account, with destinations restricted to those colleges' admission accounts which we call the set of permitted accounts, and send the quantum state to the child. The parent could deposit funds to this restricted account as usual. The child could spend the money to the permitted accounts, without the any involvement of the parent. 
We also add another feature that allows the parent to approve payments to accounts that is not one of the permitted accounts: the parent would need to digitally sign a message approving the payment to that account. 

Such a scheme is presented in \cref{fig:restricted_algorithms}. During the $\createrestrictedaccount$ algorithm, the parent specifies the set of permitted accounts $\{\permittedaccount_i\}_{i=1}^n$. 
In that process, the parent also generates a digital signature signing key $\sk$ and verification key $\vk_{parent}$. We define a $\restrictedsiginterpreter$ which gets a single message $m$ consisting of a list of accounts $(\account_1,\ldots,\account_r)$ and each account $\account_{j}$ is directed with a payment. 
We also use a script called $\restrictedoutputverifyscript$, which always approves payments to any of the $\permittedaccount$s, but requires the parent's signature to any other account. In other words, without the parent's approval, payments to addresses which are not permitted are blocked by the $\restrictedoutputverifyscript$.
\begin{figure}[thbp]
\noindent$\restrictedsiginterpreter_{\vk_{parent},\{\permittedaccount_i\}_{i=1}^n}(m)$
\begin{compactenum}
  \item Interpret $m$ as $(\account_1,\ldots,\account_r)$.
  \item For every $j \in [r]$:
  \begin{compactenum}
      \item If $\account_j \in \{\permittedaccount_i\}_{i=1}^n$, set $b_j = 1$, otherwise, set $b_j = 0$.
  \end{compactenum}
  \item Output $( \simpleoutputscript_{\account_1,\ldots,\account_r},\restrictedoutputverifyscript_{\vk_{parent},(b_1,\ldots,b_r)} ) $.
\end{compactenum}

\noindent$\restrictedoutputverifyscript_{\vk_{parent},(b_1,\ldots,b_r)}(t,\vswitness)$
\begin{compactenum}
  \item If $b_t=1$ or $\ds.\verify_{\vk_{parent}}(t,\vswitness))$, then Output $1$.
  \item Else, Output $0$.
\end{compactenum}

\noindent$\createrestrictedaccount (\bsk,\{\permittedaccount_i\}_{i=1}^n )$\hfill // If one-shot signatures are used, $\bpk$ is used instead of $\bsk$. \\
\noindent \textbf{Assumes:} Digital Signature $\ds$, and $\qs$ is an uncloneable signature scheme.
\begin{compactenum}
    \item Interpret $\bsk$ as $(\msk,\sk)$.
    \item $(\vk_{parent},\sk)\gets \ds.\keygen(\secparam)$.
    \item $(\pk, \qsk) \gets \qs.\qgen(\msk)$.\hfill // If $\qs$ is a one-shot signature scheme, $\msk$ should be replaced with $\mpk$.
    \item $\ket{\$}\equiv \qsk,\ \account \equiv (\pk,\restrictedsiginterpreter_{\vk_{parent}, \{\permittedaccount_i\}_{i=1}^n} )$.
    \item Output $ (\ket{\$},\account)$.
\end{compactenum}

\caption{Restricted account algorithms.}
\label{fig:restricted_algorithms}
\end{figure}

\subsection{Proof of reserves} 
\label{sec:proof_of_reserves}
We will focus on a simple yet non-trivial use case. A more involved setting, where an exchange proves its solvency, is discussed in ~\cite[Section 4.4]{NBF+16}, as well as privacy aspects in Ref.~\cite{DV19} and references therein.

Our (simpler) goal is the following: Alice would like to convince Bob that she has \$1000. This could be implemented using the original construction of quantum money from one-shot signatures. The protocol is as follows:
\begin{enumerate}
    \item Alice generates $(\qsk_A,\pk_A) \gets \qgen(\crs) $, and sends $\pk_A$ to Bob.
    \item Bob generates $(\qsk_B,\pk_B) \gets \qgen(\crs) $, and then $\sigma_B \gets \qsign(\qsk_A,\pk_A) $, and send $\pk_B$ and $\sigma_B) $ to Alice.
    \item Alice checks that $\verify(\crs,\pk_B,\pk_A,\sigma_B)=1$, and aborts otherwise. She then computes $\sigma_\$ \gets \qsign(\qsk_\$,\pk_B) $. Alice  sends $\pk_\$, \sigma_\$ $, as well as the chain of signatures that certify that $\qsk_\$ $ was a valid money state (see  \cpageref{par:quantum_money_from_OSS} for more details regarding the construction of quantum money from one-shot signatures).
    \item Bob verifies that $\verify(\crs,\pk_B,\pk_\$)=1$ (note that here $\pk_B$ has the role of the message which is verified) and that the chain of signatures that certify that $\qsk_\$ $ was a valid quantum money state. 
\end{enumerate}

The approach above tunnels Alice's money through Bob's account while ensuring that it won't ``get stuck''. The main drawback of this approach is that it requires that Bob would have a quantum computer. We can use a different solution using our framework that does not have that disadvantage: Bob does not need to have any quantum computing resources in this solution.

Bob sends some random challenge text $r_B \in \{0,1\}^\secpar$ to Alice. Alice would create a fresh account $\account'$, transfer her funds to the fresh account, and append $r_B$ as the auxiliary information, $\aux$, to her (signed) message. 
This is achieved by replacing \cref{it:simplepay_m} in $\simplepay$ with: Set $m \equiv (\simpleoutputscript_{\account',\ldots,\account'},\simpleoutputverifyscript,r_B)$.
Alice then sends the new account and proof that its balance is indeed \$100. Note that the last part of the witness would contain a signature of the challenge text as $\aux$.  Bob can now check that, indeed, the balance of the new account is \$100 and that $\aux$ is indeed $r_B$.


\subsection{Colored coins, smart property and tokenized securities} 
\label{sec:colored_coins_smart_property_and_tokenized_securities}
We usually think of bitcoins as fungible, meaning each bitcoin is worth just like the other. Nevertheless, Bitcoin transactions can be traced, posing a problem from a privacy perspective. Yet, this traceability is the desired property in the context of \emph{colored coins}~\cite{Ros12,NBF+16}. A colored coin can be seen as a single satoshi (which is worth $10^-8$ bitcoin, hence, has almost no value on its own) which is associated with something else, and become colored. Where could a colored coin be useful? Next, we discuss two such examples. 

A corporation could issue 100 colored coins and sell each of these colored coins in the open market as is done today in an Initial Public Offering (IPO). This colored coin provides the main functions that standard stocks offer: a) The owner of that colored coin could then trade it with others without the need for trusted exchange. b) The corporation could pay dividends in bitcoin to the current holder of the colored coin directly, without trusted intermediaries. c) The colored coin provides shareholders with voting rights. 

Another example is smart property: a colored coin may be associated with a physical item, such as a car. The car's lock could be programmed in such a way that the car opens only with a challenge-response interaction with the \emph{current} owner of the colored coin associated with that car. This way, car ownership could be done by using Bitcoin's blockchain, replacing one of the functions that (trusted) agencies such as the DMV provide. 

We can implement both of these use-cases using our quantum payment scheme. Our design does not suffer from the main draw-backs of blockchain based solutions:
\begin{enumerate}
    \item A colored coin transaction, like any other Bitcoin transaction, takes 10 minutes on average to get a single confirmation, whereas our solution has no such latency.
    \item Transacting with colored coins incurs a transaction fee, just like any other Bitcoin transaction. Our approach does not have any pecuniary cost attached.
    \item Our solution has no cap on the throughput, whereas the Bitcoin network supports only a few transactions per second, globally.
    \item Creating a stock or a smart property token using our system has no associated pecuniary cost (not even one satoshi).
\end{enumerate}
 
We first explain how to issue and use a colored coin using the quantum payment scheme. The construction is presented in full detail in \cref{fig:colored_coin_scheme,fig:colored_coin_balance}. The issuer generates the signing and verification keys of the digital signature scheme and publicly announces the verification key. The issuance of a colored coin is done by running $\createsimpleaccount$, followed by executing $\coloraccount$, where $\coloraccount$ is the same as $\topup$, except the company uses its own secret key to sign the account, instead of the bank's secret key (which, the issuer has no access to). 

Unlike an account that could hold an unbounded balance, a colored account is a boolean property: an account can be either colored or not.  We emphasize that the same account should not be colored twice. This makes the underlying design more straightforward and can be changed with some added complexity.

Payments are made similarly to the way they are done with simple accounts, with a slight change, which allows a mechanism to pay \emph{dividends}.
Recall that in general, $\outputscript(i)$ determines where the $i$'th dollar would be transferred. We make use of the fact that $\outputscript(0)$ had no meaning until now in our design and used the output-script $\coloredoutputscript(0)$ as the address, which determines where the \emph{colored coin} would be transferred. 
The $\coloredpay$ algorithm forwards all the payments made to the current account beyond $n$ to the next colored account ($\caccount$) as well, using the same approach as in \emph{permanent accounts} (see \cref{sec:permanent_accounts}). This provides a simple way to pay \emph{dividends} in the context of stocks: all dividends above $n$ will be paid to the \emph{next} holder of the colored coin, that is, the owner of $\caccount$. 

To pay dividends, the issuer would pay to the \emph{original} account to which the algorithm $\coloraccount$ has been applied. The witness of all such payments should be made publicly available so that the current and future owners can reclaim these funds. 

To verify a colored coin, the $\verifycolored$ algorithm should be used. The $\verifycolored$ is very similar to $\balance$, except for minor changes. Recall that a colored coin is a boolean property, and therefore this algorithm returns either $0$ or $1$. Additionally, it uses $\cvk$ (the issuer's verification key) to verify the signature which was produced during $\coloraccount$ (see \cref{it:colored-valid_message_is_signed} in \cref{fig:colored_coin_balance}). One of the roles of the verification is to make sure that all future payments would be transferred to the account. Suppose that until now, $n-1$ dollars were paid as dividends. This would be known to the receiver of the colored coin since all the witnesses of such payments are announced publicly. In that case, the verification would be run with the parameter $\dividends$ set to $n$. The algorithms makes sure that all future payments would be forwarded to account. This is achieved by making sure that the previous payments that were made used $\simpleoutputverifyscript$ (see \cref{it:colored-main_checks_outputverifyscript}). The fact that the $\outputverifyscript$ is $\simpleoutputverifyscript$ , which always outputs true (see \cref{fig:prudent_contract_scheme}), means that there are no restrictions that might lock the transfer (see \cref{it:main_checks_outputverifyscript} in \cref{fig:prudent_contract_scheme_balance}). 
The algorithm also makes sure that the $\outputscript$ is $\coloredoutputscript$ and that all payments above $\dividends$ are indeed forwarded to $\account$ (see \cref{it:colored_checks_ouput_script}). Note that this is also done recursively (see \cref{it:verify_color_recursive}). 

A colored coin that represents a share in a corporation could also provide \emph{voting rights}. In this case, the corporation would publicly announce a proposal. The proposal should contain a long random string, called the proposal number\footnote{The proposal number is long (rather than a sequential number that increases by 1 for every vote) so that a shareholder cannot guess that and cast a vote in advance, and then selling the share to  others.}. 
Shareholders can use their colored coin to vote by running $\coloredpay$ and moving their colored coin to a new account, and attaching $\aux$ which contains their vote:  the message to be signed in \cref{it:colored_pay_m} in \cref{fig:colored_coin_scheme} would be appended with $\aux$, where $\aux = (\text{proposal number, vote})$ (e.g., the vote might be "YES" when a proposal for a merger is made, or "Alice" when voting on a new board of directors). The payment witness, which contains $\aux$ as well as a proof that the account is indeed valid, would be sent to the corporation. The corporation would check the validity of the witness, and cast a vote according to $\aux$. In case the same share is used to vote twice, only the one which appears earlier in the chain of witnesses would be considered as valid.
 
A colored coin could also be used as smart property. Consider the example given in \cref{sec:introduction}: a car manufacturer could issue and associate a colored coin with every new car. The car could be programmed so that it issues a random challenge number and open its doors only upon receiving a signed message of that challenge text, in the same manner that the corporate verified votes in the previous example. 
Note that this scheme is only secure under the assumption that the attacker cannot undermine the physical locks and the software used in the car.
One of the main advantages of this approach is that car ownership could be proved without a centralized registry. 

\begin{figure}[thbp]
\noindent\textbf{Assumes:} Digital Signature $\ds$, Uncloneable signature $\qs$.\\
\noindent$\coloredcoinkeygen(\secparam)$
\begin{compactenum}
\item $(\csk,\cvk) \gets \ds.\keygen(\secparam)$.
\item Output $(\csk,\cvk)$.
\end{compactenum}

\noindent $\coloraccount(\csk, \account)$ \hfill // Should be compared with $\topup$ in \cref{fig:prudent_contract_scheme}. \\
Assumes that $\account$ was generated using $\createsimpleaccount$.
\begin{compactenum}
\item Let $\rand\gets_R\{0,1\}^\lambda$.
\item Set $m\equiv (\account,\rand)$.
\item $\sigma \gets$ \sout{$\ds.\sign_{\sk}(m)$} \uwave{$\ds.\sign_{\csk}(m)$}.
\item Output $\pwit \equiv (\rand,\sigma) $.
\end{compactenum}

\noindent $\coloredoutputscript_{(\account_1,\ldots,\account_n),\caccount}(i)$
\begin{compactenum}
\item If $1 \leq i \leq n$
    \begin{compactenum}
    \item Output $\account_i$.
    \end{compactenum}
    \item Else if $i=0$ or $i>n$ \hfill // $\coloredoutputscript(0)$ is used to specify to where the colored coin is transferred. The $\caccount$ is also used as a forwarding address for $i>n$, as in $\permanentoutputscript$. 
    \begin{compactenum}
        \item Output $\caccount$.
    \end{compactenum}   
\end{compactenum}

 \noindent $\coloredpay(\bpk, \ket{\$},\account, \witness,(\account_1,\ldots,\account_r),\caccount )$ \\
Assumes that $\balance(\bpk,\account,\witness)=r$.
\begin{compactenum}
    \item \label{it:colored_pay_m} Set $m \equiv (\coloredoutputscript_{(\account_1,\ldots,\account_r),\caccount},\simpleoutputverifyscript)$.
    \item $\sigma \gets \qs.\qsign(\ket{\$} ,m)$.
    \item For every $i\in \{0,\ldots,r\}$:
    \begin{compactenum}
        \item 
        Set $\pwit_i \equiv (\account, m, \sigma,i,NULL, \witness)$.
    \end{compactenum}
    \item Output $(\pwit_0,\ldots,\pwit_r)$.
\end{compactenum}  
\caption{Algorithms for the colored coin.}
\label{fig:colored_coin_scheme}
\end{figure}

\begin{figure}[thbp]
\noindent $\verifycolored(\bpk,\cvk, \account,\witness,$\uwave{$\dividends$}$)$ \hfill // Should be compared with $\balance$. 
\begin{compactenum}
    \item Interpret $\bpk$ as $(\mpk,\vk)$.
    \item Interpret $\witness$ as $(\witness_1,\ldots, \witness_v)$.
    \item \uwave{If $v > 1$, Output $\bot$.} \hfill // An account can only have 1 colored coin. 
    \item For every $i\in\{1,\ldots,v\}$: \hfill // The for loop is degenerate since $v=1$. 
    \begin{compactenum}
        \item If $\witness_i.size=2$: \hfill // A ``direct'' color witness  \label{it:colored-full-top-up-witness}
        \begin{compactenum}
            \item Interpret $\witness_i$ as $(\rand_i,\sigma_i)$. Set $m_i\equiv (\account,\rand_i)$.
            \item If \sout{$\ds.\verify_{\vk}(m_i,\sigma_i)=0$} \uwave{$\ds.\verify_{\cvk}(m_i,\sigma_i)=0$}, output $\bot$. \label{it:colored-full-top-up-contains-ds}
            \item \sout{If $m_i=m_j$ for some top-up witness $j<i$, output $\bot$.} \hfill // Redundant, since $v=1$. 
        \end{compactenum}
        \item Else:
        \begin{compactenum}
            \item Interpret $\witness_i$ as $(\account_i, \vec m_i,t_i, \vec\sigma_i,\vswitness_i, \witness'_i)$.
            \item Interpret $\account_i$ as $(\pk_{i,1},\ldots,\pk_{i,n_i},\siginterpreter_i)$. 
            \item \label{it:colored-definition_of_vec_m} Interpret $\vec m_i$ as $m_{i,1},\ldots,m_{i,n_i}$ where $m_{i,j}\in \{0,1\}^*\cup \visiblespace$ and $\vec \sigma_i$ as $\sigma_{i,1},\ldots,\sigma_{i,n_i}$ where $\sigma_{i,j}\in \{0,1\}^*\cup \visiblespace$.  
            \item Output $0$ unless the all the following tests succeed: 
            \begin{compactenum}
                \item \label{it:colored-siginterpreter_is_monotone_test} $\siginterpreter$ is monotone (\cref{def:siginterpreter}). \hfill // See \cref{rem:monotonicity_could_be_hard_to_compute}. 
                \item \label{it:colored-valid_message_is_signed} For every $j\in \{1,\ldots,n_i\}$ such that $m_{i,j} \neq \visiblespace$, $\qs.\verify(\mpk,\pk_{i,j},m_{i,j},\sigma_{i,j})=1$.
                \item \label{it:colored-interpreated_message_valid} Set $\interpretedmessage_i\gets \siginterpreter(m_{i,1},\ldots,m_{i,n_i})$. Test that $\interpretedmessage_i \neq \bot$. 
                \item Interpret $\interpretedmessage_i$ as $(\outputscript_i,\outputverifyscript_i,\aux_i)$. 
                \item \label{it:colored_checks_ouput_script}\sout{$\outputscript_i(t_i)= \account$} \uwave{$\outputscript_i$ has the form $\coloredoutputscript_{(\account_1,\ldots,\account_n),\caccount}$ and $\caccount=\account$ and $n \leq \dividends $. } \hfill // Makes sure that the colored coin is transferred to $\account$, and that all future transfers beyond $\dividends$ would be forwarded to it as well. 
                
                \item \label{it:colored-main_checks_outputverifyscript}\sout{$\outputverifyscript_i(t_i,\vswitness_i)$} \uwave{$\outputverifyscript_i$ has the form $\simpleoutputverifyscript$.} \hfill // Makes sure that future dividends  would not get locked. 
                \item \sout{$(\account_i,t_i)\neq (\account_j,t_j)$ for all $j < i$.}
                \item \label{it:verify_color_recursive} \sout{$\balance(\bpk,\account_i,\witness'_i) \geq t_i$}. \\
                \uwave{$\verifycolored(\bpk,\cvk,\account_i,\witness'_i,\dividends)$=1}. 
            \end{compactenum}
        \end{compactenum}
    \end{compactenum}
    \item \sout{Output $v$.} \uwave{Output $1$.}
\end{compactenum}
\caption{The balance algorithm for a colored coin.}
\label{fig:colored_coin_balance}
\end{figure}  


\section{A transition from Bitcoin mining}
\label{sec:transition_from_bitcoin_minig}

This section outlines a road map for transitioning away from Bitcoin mining based on quantum payment schemes. 
A less drastic proposal was given in \cite{CS20}. The underlying approach is similar: it allows users who have bitcoins to effectively burn their bitcoins and open a quantum account with the same balance instead of using our main quantum payment scheme. However, the construction in ~\cite{CS20} did not provide any of the prudent contracts the current scheme supports and, therefore, provided a mechanism to switch back to bitcoin. Even though the exact mechanism that allows users to switch their quantum account back to bitcoin could be used with the payment scheme proposal, we believe it is unnecessary since the prudent contracts provide almost all the functionality used in Bitcoin. 
The main advantage of the current approach is that mining can be eliminated. 

The rest of this section is organized as follows.
We list the motivation for such a transition in \cref{sec:bitcoin_motivation}.
The implementation details for this upgrade are given in \cref{sec:imlementing_bitcoin_upgrae}.
Our approach has some drawbacks, which we discuss in \cref{sec:drawbacks}.

\subsection{Motivation for transitioning away from Bitcoin mining}
\label{sec:bitcoin_motivation}
Bitcoin~\cite{Nak08} relies on a consensus mechanism that is based on proof-of-work. This approach has various drawbacks, including:
\begin{enumerate}
    \item It is energetically inefficient: as of March 2022, according to the \href{https://ccaf.io/cbeci/index}{Cambridge Bitcoin Electricity Consumption Index}, Bitcoin mining consumes \%0.6 of the world's electricity.
    \item The network assumes an honest majority of miners. A miner (or coalition of miners) with the majority of the mining power can perform the \%51 attack, which is known to have devastating consequences (see, e.g., ~\cite[Section 2.5]{NBF+16}).
    \item \label{it:througput} Transactions throughput is limited to 7 transactions per second.
    \item \label{it:latency} A transaction takes 10 minutes on average to confirm. 
    \item \label{it:transaction_fees} A transaction fee is required to make a payment. The transaction fee is volatile: in 2021, the average fees per transaction averaged across the day had the following statistics: average $\$9.9$, standard deviation $\$10.1$, minimum $\$1.5$, maximum $\$62.7$. 
    Furthermore, it is hard to predict the minimal fee needed to get included in a block, which causes inconveniences for the senders, receivers, and wallet developers.
    \item A user must be online to send or receive a payment. 
    \item Running a full node, which is needed to verify transactions with the best guarantees, is computationally costly: it requires a lot of disk space and decent internet bandwidth, memory, and CPU. In particular, the disk space requirements grow linearly with time. 
    \item The mining mechanism is not incentive compatible for classical miners and is known to be prone to various attacks, such as selfish mining~\cite{EG14}. This point is arguably weaker than the others
    since other works have almost closed the gap in that regard, see~\cite{PS17,ADE+22}.
    \item Quantum mining is a topic that received very little attention. It is known that the current tie-breaking rule must be modified~\cite{Sat20}. The equilibrium strategy for quantum mining is known only in very few cases~\cite{LRS19}. Furthermore, for the few cases for which the equilibrium is known, a miner needs to know the other miners' hashing power to decide on their strategy, and it is not clear how to reveal this information. This is in sharp contrast with the current honest mining strategy, which is simpler and independent of the other miners' hashing power.  
    On the positive side, it is known that the so-called "Bitcoin backbone" satisfies some important guarantees in the presence of a single quantum mining adversary, see Ref.~\cite{CGK+20}. Of course, the honest majority of the hashing power assumption is adjusted to consider Grover's speed-up of the (quantum) adversary. 
\end{enumerate}
To summarize the points above, mining has severe drawbacks in classical settings, and the implications of quantum mining could potentially make it even worse. 

The Bitcoin lightning network~\cite{PD15} is a second layer network that aims to improve \cref{it:througput,it:latency,it:transaction_fees}. However, it is only an \emph{improvement} rather than a solution. This is because a user wishing to make a lightning network payment has first to open a lightning channel, which is an on-chain transaction that suffers from the same limitations as mentioned in \cref{it:througput,it:latency,it:transaction_fees}. The same holds for closing a channel. Furthermore, the security of the lightning network has also been questioned; see ~\cite{RMT19,PRH+20,MZ21,TYM+21} and references therein.

The main alternative to Proof of Work (PoW) is called Proof-of-Stake (PoS); for more details, see, e.g., ~\cite{KRD+17,GHM+17}. The main advantage is the energy efficiency and the irrelevance of quantum mining in this setting; the other drawbacks are still present, though various trade-offs can be made (for example, improved throughput and latency, at the price of a higher computational cost). 

Coladangelo and Sattath~\cite{CS20} presented a hybrid approach in which a proof-of-work-based cryptocurrency could use a quantum lightning with bolt-to-certificate as a way to upgrade Bitcoin. The way it is done is so that users can switch back and forth between the current mode in which bitcoins are stored via a quantum money-based solution. The primary motivation for going back to Bitcoin was that some forms of smart contracts were impossible to achieve in that framework. However, this approach requires maintaining a consensus mechanism and, therefore, could not be viewed as an alternative to PoW or PoS. 

In the rest of this section, we present and discuss a fully quantum alternative. Since our construction provides the vast majority of the types of transactions that are used in Bitcoin, we think it could \emph{abolish} the need to maintain mining altogether. The approach, which will be presented shortly, does not suffer from \emph{any} of the drawbacks in the list above.

\subsection{Implementing the upgrade}
\label{sec:imlementing_bitcoin_upgrae}
We now explain how such an approach could be implemented. The most crucial ingredient needed would be a one-shot signature scheme, for which we do not have candidate constructions in the plain model; see \cref{sec:drawbacks} for further discussion. 
\paragraph{Generating a \crs.}
Recall that a one-shot signature needs a common random string (\crs). Bonneau, Clark, and Goldfeder~\cite{BCG15} describe how to generate that randomness using the Bitcoin mining process. This randomness could be generated from early blocks, which can rule out specific attack vectors, such as a coalition of malicious miners trying to manipulate the \crs.

\paragraph{Replacement of the bank's digital signature.}
In our quantum payment scheme based on one-shot signatures, we assume that a trusted party, referred to as the bank, holds the digital signature signing key $\sk$ and everyone else has the verification key $\vk$. In our setting, the blockchain replaces the role of the bank. 

\paragraph{Upgrading bitcoins to a quantum account.}
Users would have ample time (say, five years) to upgrade their bitcoin to the quantum payment scheme. The user would run \createsimpleaccount, using the \crs. The \topup is done by \emph{burning} bitcoins and creating an account with the appropriate balance. A new Bitcoin opcode would be introduced, called \optopup. The \optopup would effectively burn the bitcoins and associate the same value to the account specified by the user. Of course, like any other transaction that spends those bitcoins, the transaction is signed by the user's signing key and, therefore, cannot be forged.

\paragraph{Sending a payment.} The account owner could spend the funds as usual: the receiver would create an account (if she doesn't have one already) using $\createsimpleaccount$ and send the $\account$ to the sender. The sender would use $\simplepay$ and send the payment witness $\pwit$ back to the sender.

\paragraph{Receiving a payment before the terminating block.} In the next paragraph, we will discuss the terminating block. To verify a payment before the terminating block, the receiver must access the blockchain. The receiver would use the $\balance$ algorithm, except the verification of a $\topup$ witness would be different: in this case, the verification is done by ensuring that the \optopup transaction is valid---i.e., it is spent from a Bitcoin address with a sufficient balance, this transaction is part of the longest chain, and has received many confirmations---as required from any other Bitcoin transaction. In this transition period, verification requires that the receiver has access to the blockchain and therefore needs to be online; this is not necessary after the terminating block has been mined.

\paragraph{Terminating block.}
As part of the upgrade, a block height for a terminating block\footnote{The first Bitcoin block is often called the genesis block. The terminating block would be the last block in the blockchain.} would be coordinated by the same mechanism of Bitcoin soft-forks. This last block would contain the root of a Merkle tree containing all the accounts for which \optopup has occurred, ordered by the original appearance on the Bitcoin blockchain. The terminating block would be considered valid only if this Merkle root is valid; note that the users can verify that the Merkle root is correct and reject it otherwise.

Special care needs to be made to have a consensus regarding which block is the terminating block. It is desired that a vast majority of the network agree upon this block, with high probability. A fork is not preferred because different nodes may disagree regarding which block is the terminating block. One way to achieve that is through a "difficulty bomb"---a technique that was used in Ethereum for different reasons. In a difficulty bomb, the difficulty gradually increases. An increase in difficulty by a multiplicative factor $f$ would increase the time between blocks by that same factor and decrease the probability that honest miners would create a fork by the inverse factor. Since this upgrade is of theoretical nature at this point, we leave the problem of the optimal way to agree on the terminating block for future research.  

\paragraph{Receiving a payment after the terminating block.}
The terminating block contains a Merkle root of all the \optopup transactions. This fact can be used to simplify the $\balance$ algorithm compared to a payment done before the terminating block. A top-up witness can contain a Merkle path from the relevant \optopup transaction to the Merkle root. The main advantage of this approach is that it means that the transaction is locally verifiable: only the sender and the receiver are involved in a transaction. 

\subsection{Drawbacks}
\label{sec:drawbacks}
The transitioning away from mining discussed in this section has the following drawbacks. Primarily for reasons related to hardware (see the paragraph below), this should be viewed as a long-term solution that would take many years to implement.

\paragraph{Quantum hardware.}
The suggested solution would require users to have access to quantum hardware to generate, store and measure the quantum states in the algorithms. No such hardware is commercially available today and would probably take many years to be realized.

A user may use a quantum service provider to store their quantum money. The user could take the role of the parent in the restricted account, as presented in  (see \cref{sec:restricted_accounts}) so that the service provider would not be able to steal the user's funds.

\paragraph{No construction of one-shot signatures in the standard model.} Our approach for the transition from Bitcoin mining crucially relies on one-shot signatures. 
As was mentioned in \cref{sec:introduction}, the construction in Ref.~\cite{AGKZ20} is relative to a classical oracle, whereas the recent construction in Ref.~\cite{DS22} is based on a non-collapsing hash function, which is not known to exist. A prerequisite for this transition is a construction of $\oss$ in the plain or random oracle models. 
\paragraph{Restricted forms of multi-signatures}
As discussed in \cref{sec:backup_and_multiple_signatures}, losing access to funds and embezzlement are serious risks that are partially mitigated using multi-signatures. The quantum payment scheme allows only a restricted set of multi-signatures. For example, our solution allows 2-out-of-3 multi-signatures but not 1-out-of-2 multi-signatures. Note that this approach can also reduce the risk of lost funds, or even extortion attempts, by a quantum service provider. For example, a user without quantum hardware could use a 2-out-of-3 multi-signatures account with three quantum service providers. As long as no more than one of the service providers fails (recall that the service provider cannot steal the user's funds, though it may try to extort the user), the user's funds are safe. 

\paragraph{Incompatibility with existing applications.}
A few protocols and services would cease to work if such a transition occurred.
These would have to use a different solution, such as a transition to a network dedicated to smart contracts (e.g., Ethereum): our solution does not provide an alternative to a system such as Ethereum, where various contracts require a consensus mechanism are used. 
We argue that the benefits---most notably the energetic aspects---outweigh the disadvantages, especially since these protocols and services do not constitute a large part of the Bitcoin economy and could easily migrate to another blockchain. We now list these main protocols and services that we are aware of:
\begin{enumerate}
    \item Time-stamping. Here, a service provider batches multiple documents and posts the Merkle root of these documents at regular intervals on the blockchain. For details, see~\cite[Section 9.1]{NBF+16}. The main service provider is \url{https://opentimestamps.org/} where details regarding how this service works can be found at \url{https://petertodd.org/2016/opentimestamps-announcement#how-opentimestamps-works}.
    \item \label{it:atomic-swaps} Cross-chain atomic swaps. Recall that in the secure exchange that was mentioned in \cref{sec:secure_exchange}, a malicious party can lock the funds of an honest user. A cross-chain atomic swap provides a better guarantee (see, e.g.,~\cite[Section 10.5]{NBF+16}): a malicious party cannot even lock the funds of the honest party for more than a few blocks. Cross-chain atomic swaps between various pairs of cryptocurrencies have been implemented. For example, as of March 2022, the software in the GitHub repository \url{https://github.com/decred/atomicswap} declares that it supports cross-chain atomic swaps between 10 different cryptocurrencies. 
    \item Fair online lotteries. For details, see ~\cite[Section 9.3]{NBF+16}.
    \item Randomness beacon. Bitcoin can be used as a randomness beacon, though powerful miners can manipulate the randomness~\cite[Section 9.4]{NBF+16}.
\end{enumerate}
\paragraph{Privacy and anonymity aspects.} Standard Bitcoin transactions do not guarantee privacy or anonymity for the users; see, e.g.,~\cite{RS13,RS14,MPJ+16}. There are some countermeasures to provide better anonymity that can be used in Bitcoin. Also, there are various other techniques to improve privacy and anonymity that are used in other cryptocurrencies and incompatible with Bitcoin---see ~\cite{FHZ+19} for a survey. 

On the one hand, with our approach, there is no ledger that allows everyone to see all the past transactions easily. On the other hand, a payment from an account reveals its history. 

While some approaches, such as CoinJoin---a decentralized ``mixer'', see~\cite{Max13}---seem to be incompatible with our techniques, it appears that zero-knowledge techniques could be used to improve the privacy and anonymity of the users; see also \cref{it:untraceability} in \cref{sec:discussion}.

\section{Discussion}
\label{sec:discussion}
Our construction has room for improvements, though implementing these would probably increase its complexity:
\begin{enumerate}
    \item Like cash, and many cryptocurrencies such as Bitcoin, our design is not infinitely divisible. 
    \item  We do not attempt to support more customizable key-management solutions, such as \emph{covenants}~\cite{MES16}.
    \item The witness sizes are not optimized and could become quite large. 
    \item \label{it:untraceability} We ignore aspects related to untraceability, which has been studied extensively in related contexts such as quantum money~\cite{MS10,JLS18,BS20}, and electronic cash~\cite{Cha85,CFN88}. Some of the generic approaches, such as non-interactive zero-knowledge proofs ~\cite{BFM88} and its variants, such as Zero-Knowledge Succinct Non-interactive ARgument of Knowledge (ZK-SNARK)~\cite{BCTV14}, which were proved useful in the context of cryptocurrencies (see, e.g.,~\cite{BCG+14}) seem to have the largest potential to be useful in our context as well.
\end{enumerate}
\subsection*{Acknowledgments}

\ifnum \anon=0
We wish to thank Zvika Brakerski and Isaiah Hull for valuable discussions, and Amit Behera, for his comments.
This work was supported by the Israel Science Foundation (ISF) grant No. 682/18 and 2137/19 and
by the Cyber Security Research Center at Ben-Gurion University.
\fi
\ifnum\sigconf=1
    \bibliographystyle{ACM-Reference-Format}
\else

    \bibliographystyle{alphaabbrurldoieprint}

\fi
{\footnotesize\bibliography{main}}
\appendix
\ifnum \shownomenclature=1
\printnomenclature[1in]
\fi

\section{Tokenized and semi-quantum tokenized signatures unforgeability}
\label{sec:tokenized_signatures_unforgeability}

The original definition of tokenized signatures~\cite{BS16,CLLZ21} does not associate a public key with the quantum signing key. The deletions of the public keys, which are the only changes from the syntax and correctness in \cref{def:tss}, are marked as \sout{$\pk$} in the following definition. 
 
\begin{definition}[Tokenized signatures without public keys~\cite{BS16}]
A tokenized signature scheme consists of 4 algorithms, \setup (\ppt), \tokengen\ (\qpt), \qsign\ (\qpt), and \verify\ (deterministic polynomial time) with the following syntax:
\begin{enumerate}
\item $(\mpk,\msk) \gets \setup(\secparam)$: takes a security parameter and outputs a classical master public key \mpk, and a classical master secret \msk.
\item \sout{$\pk$,}$\qsk \gets \tokengen(\msk):\qsk$ takes the master secret key $\msk$ and outputs \sout{a classical public key $\pk$, and} a quantum signing key $\qsk$.
\item $\sigma \gets \qsign(\qsk,m)$: takes a quantum signing key $\qsk$, and a classical message $m$, and outputs a classical signature $\sigma$.
\item $b \gets \verify(\mpk,$\sout{$\pk$,} $m,\sigma)$: receives a classical master public key\sout{, a classical public key,} a classical message $m$ and an alleged classical signature and either accepts or rejects.
\end{enumerate}
\paragraph{Correctness.} If $(\mpk,\msk)\gets \setup(\secparam)$ and (\sout{\pk,}$\qsk)\gets \tokengen(\msk)$ then for any message $m\in\{0,1\}^*$, $\verify(\mpk$\sout{,\pk},$m,\qsign(\qsk,m))=1$ with overwhelming probability. 
\label{def:tss_without_public_keys}
\end{definition}
To define unforgeability without public keys, the algorithm $\verify_{\ell,\mpk}(\cdot)$ is introduced. This algorithm takes as an input $\ell$ pairs $(m_{1},\sigma_{1}),\ldots,(m_{\ell},\sigma_{\ell})$ and accepts if and only if:
\begin{enumerate}
    \item \label{it:verify_k_1} All the messages are distinct, i.e. $m_i\neq m_j$  for every $1 \leq i \neq j \leq \ell$.
    \item \label{it:verify_k_2} All the pairs pass the verification test: for all $i\in [\ell]$, $\verify_{\mpk}(m_{i},\sigma_i)=1$.
\end{enumerate}
\begin{definition}[Unforgeability for tokenized signatures without public keys~\cite{BS16}]
\label{def:strong_unforgeability_without_public_keys}
A \ts scheme without public keys is unforgeable if for every $\ell=\poly$ a \qpt adversary cannot sign $\ell+1$ distinct messages by using the public key and $\ell$ signing tokens:
\begin{equation}
\Pr\left[
\begin{array}{c}
(\msk,\mpk)\gets \setup(\secparam) \\ 
\qsk_1 \gets \tokengen_{\msk}\\
\vdots \\
\qsk_\ell \gets \tokengen_{\msk}\\
\end{array}
:\verify_{\ell+1,\mpk}( \qadv(\mpk,1^\kappa,\qsk_1 \tensor \ldots \tensor \qsk_{\ell})) = 1 \right]\leq \negl.
\end{equation}
\end{definition}
\begin{remark}
We can similarly define semi-quantum tokenized signatures without public keys by replacing the $\qgen$ algorithm in \cref{def:tss_without_public_keys} with a classical communication protocol $\qgen\equiv\langle \qgen.\sen,\qgen.\rec\rangle$ for generation of the quantum signing key $\qsk$ (similar to how it is done in \cref{def:semi-quantum-tokenized-signatures}). 
The correctness and unforgeability definitions can be adapted accordingly.
\end{remark}

It turns out that the definition above is insufficient for our purposes. In our case, the quantum signing keys are non-equivalent, as each such quantum signing token may have a different value associated with it. 
 
In Ref.~\cite{BS16} and Ref.~\cite{Shm22}, the notions of a tokenized signature mini-scheme and a semi-quantum tokenized signature mini-scheme were introduced respectively, which are analogous to a quantum money mini-scheme~\cite{AC13}. 
\begin{definition}[Tokens for digital signatures mini-scheme, adapted from~\cite{BS16}]
A tokenized signature mini-scheme\footnote{In \cite{BS16} mini-scheme unforgeability was called one-time unforgeability, and the syntax was slightly different. To avoid confusion between one-time and one-shot (which is used to refer to one-shot signatures, which is indeed very different), we avoid using the term one-time unforgeability.} consists of 3 algorithms, \tokengen\ (\qpt), \qsign\ (\qpt), and \verify\ (deterministic polynomial time) with the following syntax:
\begin{enumerate}
\item $\tokengen(\secparam)$ generates a classical public key $\pk$, and a quantum signing token $\qsk$. 
\item $\qsign$ receives a quantum state (presumably, a signing token), and a message $m$, and outputs a classical signature $\sigma$.
\item $\verify$ receives a public key $\pk$, a message, and an alleged signature and either accepts or rejects.
\end{enumerate}
\paragraph{Correctness.} If $(\pk,\qsk)\gets \tokengen(\secparam)$ then $\verify_\pk(m,\qsign(\qsk,m))=1$ for any message $m\in\{0,1\}^*$ with overwhelming probability. 
\paragraph{Unforgeability.} For any \qpt adversary $\qadv$, there is a negligible function $\epsilon$ such that for all $\secpar$:
\begin{equation}
\Pr\left[
\begin{array}{c}
    (\pk,\qsk)\gets \tokengen(\secparam) \\
    \{(m_0,\sigma_0),(m_1,\sigma_1) \} \gets \qadv(\pk,\qsk)
\end{array}
:
\begin{array}{c}
    \verify^2(\mpk,\pk,m_0,m_1,\sigma_0,\sigma_1)=1
\end{array}
 \right]\leq \epsilon(\secpar).
\end{equation}
\label{def:tss_mini_scheme}
\end{definition}

\begin{definition}[Semi-quantum Tokens for digital signatures mini-scheme, adapted from~\cite{Shm22}]
A semi-quantum tokenized signature mini-scheme consists of two algorithms $\qsign\ (\qpt)$, and $\verify\ (\ppt)$ with the same syntax as a tokenized signature mini-scheme. The $\qgen$ algorithm is replaced with the $\qgen$ protocol with classical communication with the following syntax:

$\langle\qgen.\sen,\qgen.\rec\rangle$ is a classical-communication protocol between two parties, a $\ppt$ sender $\qgen.\sen$ which gets the master secret key $\msk$ as input and a $\qpt$ receiver $\qgen.\rec$ (without any input), and generates a classical public key $\pk$, and a quantum signing token $\qsk$. 
\paragraph{Correctness.} If $(\pk,\qsk)\gets \langle\qgen.\sen,\qgen.\rec\rangle_{(\out_\sen,\out_\rec)}$ then $\verify_\pk(m,\qsign(\qsk,m))=1$ for any message $m\in\{0,1\}^*$ with overwhelming probability. 
\paragraph{Unforgeability.} For any \qpt adversary $\qadv$, there is a negligible function $\epsilon$ such that for all $\secpar$:
\begin{equation}
\Pr\left[
\begin{array}{c}
     (\pk,m_0,m_1,\sigma_0,\sigma_1 ) \gets \langle\qgen.\sen,\qadv\rangle_{\out_\sen,\out_\qadv}
\end{array}
:
\begin{array}{c}
    \verify^2(\mpk,\pk,m_0,m_1,\sigma_0,\sigma_1)=1
\end{array}
 \right]\leq \epsilon(\secpar).
\end{equation}
\label{def:semi-quantum-tss_mini_scheme}
\end{definition} 

We define a stronger variant of unforgeability for tokenized signatures with public keys (i.e., as defined in \cref{def:tss}). 
Recall that the unforgeability property in \cref{def:tss} guarantees that an adversary cannot generate two valid signatures associated with the same public key.
We strengthen the definition by adding the requirement that if the adversary did not receive the public key $\pk$ from the oracle access to $\qgen$, then the adversary cannot produce even a single valid signed message associated with $\pk$.
Formally, the following condition is added to the one in \cref{eq:ts_unforgeability}:
\begin{definition}
For a \qpt adversary $\qadv$ with oracle access to $\qgen(\msk)$, we define the set PKS as the set of public keys generated by the oracle and given to the adversary. For any \qpt adversary $\qadv$, there is a negligible function $\epsilon$ such that for all $\secpar$:
\begin{equation}
\Pr\left[
\begin{array}{c}
    (\mpk,\msk)\gets \setup(\secparam)\\
    (\pk, m,\sigma) \gets \qadv^{\tokengen(\msk)}(\mpk)
\end{array}
:
\begin{array}{c}
    \verify(\mpk,\pk,m,\sigma)=1 \text{ and }\\
    \pk \notin PKS
\end{array}
 \right]\leq \epsilon(\secpar).
\label{eq:ts_variant_unforgeability}
\end{equation}
\label{def:variant_unforgeability}
This security notion can be extended to the semi-quantum tokenized signatures by replacing the oracle access to $\qgen(\msk)$ oracle with a oracle access to the classical protocol $\langle \qgen.\sen(\msk),\cdot\rangle$ as it is done in \cref{def:semi-quantum-tokenized-signatures}. 
\end{definition}
 
Ben-David and Sattath~\cite{BS16} showed a black-box construction of a (standard) tokenized signature without public keys based on a tokenized signature mini-scheme. For completeness, we present their construction in \cref{fig:full_scheme_from_mini_scheme}. 
 Here, we show that the resulting (standard) tokenized signature is also a tokenized signature with public keys and that it satisfies the unforgeability defined in \cref{def:tss}, and also the additional property defined in \cref{def:variant_unforgeability}. The same techniques can be used to lift a semi-quantum tokenized signature mini-scheme to a (standard) semi-quantum tokenized signature. The property in \cref{def:variant_unforgeability} is not needed in this work, but might be of independent interest.

\begin{figure}[thbp]
\noindent\textbf{Assumes:} $\tsms$ is a tokenized signature mini-scheme, $\ds$ is a digital signature scheme.\\
\noindent$\setup(\secparam)$
\begin{compactenum}
\item $(\mpk,\msk) \gets \ds.\keygen(\secparam)$.
\item Output $(\mpk,\msk)$.
\end{compactenum}

\noindent $\qgen(\csk)$
\begin{compactenum}
\item $(\pk,\qsk') \gets \tsms.\qgen(\secparam)$.
\item $\rho \gets \ds.\sign_{\msk}(\pk')$.
\item Output $(\pk,\qsk=(\qsk',\rho))$.
\end{compactenum}

\noindent $\qsign((\qsk',\rho),m)$
\begin{compactenum}
\item $\tau \gets \tsms.\qsign(\qsk,m)$.
\item Output $\sigma\equiv (\rho,\tau)$.
\end{compactenum}

 \noindent $\verify(\mpk,\pk,m,\sigma)$
\begin{compactenum}
    \item Interpret $\sigma$ as $(\rho,\tau)$. 
    \item Output $\ds.\verify_\mpk(\pk,\rho) \wedge \tsms.\verify(\pk,m,\tau)$.
\end{compactenum}  
\caption{A standard tokenized signature scheme from a tokenized signature mini-scheme.}
\label{fig:full_scheme_from_mini_scheme}
\end{figure}
\begin{proposition}
If $\tsms$ is an unforgeable tokenized signature mini-scheme (see \cref{def:tss_mini_scheme} and $\ds$ is a PQ-EU-CMA digital signature (see \cref{def:digital_signature}), then the construction in \cref{fig:full_scheme_from_mini_scheme} is an unforgeable tokenized signature (see \cref{def:tss}), which also satisfies the additional unforgeability property in \cref{def:variant_unforgeability}.
\label{prop:lifting_theorem_ts}
\end{proposition}
\begin{proof}[Sketch proof] The proof is the same as in \cite{BS16}, but we will give a sketch proof for completeness. 
Suppose the scheme does not satisfy \cref{def:variant_unforgeability}. In that case, the same adversary can be used to break the digital signature scheme: note that the adversary can produce a valid signature for a message that she has not given (we know that this message was not given to the adversary by the criterion $\pk \notin PKS$ in the definition). This is, of course a contradiction to our hypothesis that $\ds$ is  PQ-EU-CMA.

An adversary which violates \cref{eq:ts_unforgeability}, but not \cref{def:variant_unforgeability} can be used to construct an adversary which violates the mini-scheme unforgeability. The main difference is that here a mini-scheme adversary receives only one public key and quantum signing token. In contrast, the adversary has oracle access to the $\qgen$ oracle in the full scheme. To break the mini-scheme security, the mini-scheme adversary could generate the digital signature keys, as well as all the other tokens except one chosen at random, and simulate the full-scheme adversary. Let $n$ be the number of calls made to the $\qgen$ oracle by the full-scheme adversary. Since the full-scheme adversary's view remains exactly the same,  there is $\frac{1}{n}$ probability, that the public key which would be broken is the one which was given to the mini-scheme adversary, and therefore, if the full scheme's adversary success probability is non-negligible, then so is the mini-scheme's success probability. 
\end{proof}
A similar construction and the same proof technique can be used to perform the lift in the semi-quantum regime, but we omit the formal proof for brevity:

\begin{proposition}\label{prop:lifting_theorem_sqts}
Assuming the existence of PQ-EU-CMA digital signatures, any unforgeable semi-quantum tokenized signature mini-scheme (see \cref{def:semi-quantum-tss_mini_scheme}) can be lifted to an unforgeable tokenized signature scheme (see \cref{def:semi-quantum-tokenized-signatures}), which also satisfies the additional unforgeability property in \cref{def:variant_unforgeability}.

\end{proposition}
 \begin{theorem}[{\cite[Corollary 1]{CLLZ21}}] Assuming post-quantum indistinguishability obfuscation, and one-way function, there exists a strongly unforgeable tokenized signature mini-scheme. 
  \label{thm:tokenized_signatures_mini-scheme_exist}
 \end{theorem}
 
 \begin{theorem}[{\cite[Theorem 1.1]{Shm22}}]
 Assume that Decisional LWE has sub-exponential quantum indistinguishability and that
indistinguishability obfuscation for classical circuits exists with security against quantum polynomial
time distinguishers. Then, there is an unforgeable semi-quantum tokenized signature mini-scheme.
\label{thm:semi_quantum_tokenized_signatures_mini-scheme_exist}
 \end{theorem}
\begin{proof}[Proof of \cref{thm:tokenized_signatures_exist}]\label{pf:thm:tokenized_signatures_exist}
The assumptions made in \cref{thm:tokenized_signatures_exist} are the same as required in \cref{thm:tokenized_signatures_mini-scheme_exist}.
It is known that post-quantum one-way functions imply PQ-EU-CMA digital signatures (see \cref{def:digital_signature,def:unforgeability_of_digital_signature}), see Ref.~\cite[Theorem 5.4]{Son14}. Hence, the assumptions in \cref{thm:tokenized_signatures_exist} imply the requirements in \cref{prop:lifting_theorem_ts}. Therefore, by combining \cref{prop:lifting_theorem_ts}, and \cref{thm:tokenized_signatures_mini-scheme_exist}, we get the desired result. 
\end{proof} 

\begin{proof}[Proof of \cref{thm:semi-quantum_tokenized_signatures_exist}]\label{pf:thm:semi-quantum_tokenized_signatures_exist}
The assumptions made in \cref{thm:semi-quantum_tokenized_signatures_exist} are the same as required in \cref{thm:semi_quantum_tokenized_signatures_mini-scheme_exist}.
Note that, sub-exponential quantum indistinguishability of Decisional LWE (which is stronger than vanila post-quantum indistinguishability of DLWE) implies post-quantum one-way functions, see Ref.~\cite{Reg05}. As noted in the proof of \cref{thm:tokenized_signatures_exist}, post-quantum one-way functions imply PQ-EU-CMA digital signatures (see \cref{def:digital_signature,def:unforgeability_of_digital_signature}), see Ref.~\cite[Theorem 5.4]{Son14}. Hence, the assumptions made in \cref{thm:semi-quantum_tokenized_signatures_exist} imply the requirements in \cref{prop:lifting_theorem_sqts}. Therefore, by combining \cref{prop:lifting_theorem_sqts}, and \cref{thm:semi_quantum_tokenized_signatures_mini-scheme_exist}, we get the desired result. 
\end{proof} 
   
\end{document}